\pdfoutput=1
\newif\ifarxiv
\arxivtrue

\ifarxiv
\documentclass[acmsmall,nonacm,screen]{acmart}
\else
\documentclass[acmsmall,screen,review]{acmart}
\fi

\ifarxiv
\setcopyright{cc}
\setcctype[4.0]{by}
\else
\renewcommand\footnotetextcopyrightpermission[1]{}

\setcopyright{rightsretained}
\acmYear{2024}
\copyrightyear{2024}
\acmJournal{PACMPL}
\acmNumber{POPL}

\fi

\usepackage{xcolor,soul}
\usepackage[noline,noend]{algorithm2e}
\usepackage{listings}
\usepackage[shortlabels]{enumitem}
\usepackage{mathpartir}
\usepackage{utf8-symbols}
\usepackage[nounderscore]{syntax}
\usepackage[capitalise,noabbrev]{cleveref}

\theoremstyle{definition}
\newtheorem{definition}{Definition}[section]

\newcommand{\concat}{%
  \mathbin{{+}\mspace{-6mu}{+}}%
}

\newcommand{\terminator}{
  \otimes
}

\newcommand{\emptydelta}{
  \varnothing
}

\newenvironment{DIFnomarkup}{}{}
\newcommand{\diffblock}[1]{#1}

\SetKwBlock{Loop}{loop}{end}

\AtBeginDocument{%
  \providecommand\BibTeX{{%
    \normalfont B\kern-0.5em{\scshape i\kern-0.25em b}\kern-0.8em\TeX}}}

\newif\ifcomments
\commentstrue
\ifcomments
    \providecommand{\shadaj}[1]{{\protect\color{brown}{\bf [shadaj: #1]}}}
    \providecommand{\conor}[1]{{\protect\color{red}{\bf [conor: #1]}}}
    \providecommand{\alvin}[1]{{\protect\color{purple}{\bf [alvin: #1]}}}
    \providecommand{\mae}[1]{{\protect\color{blue}{\bf [mae: #1]}}}
    \providecommand{\joe}[1]{{\protect\color{teal}{\bf [joe: #1]}}}
    \providecommand{\jmh}[1]{{\protect\color{teal}{\bf [joe: #1]}}}
    \providecommand{\david}[1]{{\protect\color{green}{\bf [david: #1]}}}
    \providecommand{\chris}[1]{{\protect\color{violet}{\bf [chris: #1]}}}
    \providecommand{\davidmwei}[1]{{\protect\color{pink}{\bf [david wei: #1]}}}
    \providecommand{\kaushik}[1]{{\protect\color{orange}{\bf [kaushik: #1]}}}
    \providecommand{\justin}[1]{{\protect\color{green}{\bf [justin: #1]}}}
    \providecommand{\mingwei}[1]{{\protect\color{rhodamine}{\bf [mingwei: #1]}}}
    \providecommand{\rithvik}[1]{{\protect\color{red}{\bf [rithvik: #1]}}}
    \providecommand{\nc}[1]{{\protect\color{pink}{\bf [nc: #1]}}}
     \providecommand{\accheng}[1]{{\protect\color{olive}{\bf [accheng: #1]}}}
    \providecommand{\dan}[1]{{\protect\color{purple}{\bf [dan: #1]}}}
\else
    \providecommand{\shadaj}[1]{}
    \providecommand{\conor}[1]{}
    \providecommand{\alvin}[1]{}
    \providecommand{\mae}[1]{}
    \providecommand{\joe}[1]{}
    \providecommand{\jmh}[1]{}
    \providecommand{\david}[1]{}
    \providecommand{\chris}[1]{}
    \providecommand{\davidmwei}[1]{}
    \providecommand{\kaushik}[1]{}
    \providecommand{\justin}[1]{}
    \providecommand{\mingwei}[1]{}
    \providecommand{\rithvik}[1]{}
    \providecommand{\nc}[1]{}
    \providecommand{\accheng}[1]{}
    \providecommand{\dan}[1]{}
\fi

\begin{document}

\title{Flo: a Semantic Foundation for Progressive Stream Processing}

\author{Shadaj Laddad}
\affiliation{%
 \institution{UC Berkeley} \country{USA}}
 \email{shadaj@cs.berkeley.edu}

\author{Alvin Cheung}
\affiliation{%
 \institution{UC Berkeley} \country{USA}}
 \email{akcheung@cs.berkeley.edu}

\author{Joseph M. Hellerstein}
\affiliation{%
 \institution{UC Berkeley} \country{USA}}
 \email{hellerstein@cs.berkeley.edu}

\author{Mae Milano}
\affiliation{%
 \institution{Princeton University} \country{USA}}
 \email{mpmilano@cs.princeton.edu}
    
\begin{abstract}
Streaming systems are present throughout modern applications, processing continuous data in real-time. Existing streaming languages have a variety of semantic models and guarantees that are often incompatible. Yet all these languages are considered ``streaming''---what do they have in common? In this paper, we identify two general yet precise semantic properties: streaming progress and eager execution. Together, they ensure that streaming outputs are deterministic and kept fresh with respect to streaming inputs. We formally define these properties in the context of Flo, a parameterized streaming language that abstracts over dataflow operators and the underlying structure of streams. It leverages a lightweight type system to distinguish bounded streams, which allow operators to block on termination, from unbounded ones. Furthermore, Flo provides constructs for dataflow composition and nested graphs with cycles. To demonstrate the generality of our properties, we show how key ideas from representative streaming and incremental computation systems---Flink, LVars, and DBSP---have semantics that can be modeled in Flo and guarantees that map to our properties.
\end{abstract}


\maketitle

\section{Introduction}
Stream processing is an increasingly important component of modern applications, from real-time analytics to collaborative tools. These applications must respond with low latency to events as they arise and often process long streams of data. Furthermore, these applications often involve stateful processing, where the output of a computation depends on the history of the inputs.

Many streaming applications are expressed as dataflow programs~\cite{dataflow_model}, specified as a directed graph of operators. Each node is an operator that consumes and produces data elements, and the edges represent the flow of data between them. This model is used in systems like Apache Flink~\cite{flink}, Spark~\cite{discretized_streams}, StreamIt~\cite{streamit}, and many functional-reactive programming languages~\cite{paykin2016essence}. Dataflow programs benefit from being written in a declarative manner that abstracts away from low-level details such as how operators are scheduled and where state in the system is accumulated~\cite{noria,cql,adaptive_stream_processing,borealis}. This makes it easy for compilers to optimize dataflow programs, since they can rearrange and transform operators within the graph without affecting the observable behavior of the program.

Existing streaming languages present a variety of semantics and aim to provide various guarantees. But several streaming languages do not even agree on what constitutes a stream! They can be ordered sequences~\cite{flink, streamit}, or sets~\cite{dedalus}, or even lattices~\cite{lvars} or Z-Sets~\cite{dbsp}. These languages also vary in their semantics for state persistence, and offer a range of approaches for concepts like windowed aggregations and batched execution. But they also have much in common: streaming languages tolerate changing inputs and aim to produce outputs as early as possible. Yet these ideas have remained fuzzy and tied to incompatible semantics.

In this paper, we distill these common traits into two key properties: \textbf{streaming progress} and \textbf{eager execution}. We formally define these properties in the context of \textbf{Flo}, a parameterized streaming language that accommodates a range of streaming semantics while providing sufficient structure to precisely define our proof objectives. Flo abstracts away from notions of underlying collection types, such as ordered sequences, and supports semantics that many streaming languages cannot reason about~\cite{streamtypes}, such as retractions.

A key challenge in streaming systems is ensuring that the program makes \emph{progress}. Unlike traditional languages, the definition of progress in streaming languages has long remained fuzzy and tied to very specific semantics. In Flo, we introduce a \textbf{general yet precise} formal definition called \textbf{streaming progress}, which uses \emph{stream termination} (inspired by work from the databases community~\cite{tucker2003exploiting}) as a common semantic feature to make guarantees about streaming outputs. Streaming progress guarantees that a Flo program produces \textbf{as much output} as possible given its input, and that the program \textbf{will not block} on a stream that may never terminate.

To enforce streaming progress, we introduce a \emph{lightweight} type system that differentiates between bounded and unbounded streams. Bounded streams are guaranteed to eventually terminate, while unbounded streams may never terminate. Operators can only block on bounded streams, and must always make progress with respect to unbounded streams. These lightweight types can be layered on arbitrary underlying collection types, such as Stream Types~\cite{streamtypes}, sets, or even lattices.%

Where streaming progress focuses on ensuring that outputs are produced in a timely fashion relative to inputs, \textbf{eager execution} ensures that the outputs are \emph{deterministic}. Many streaming systems make strong assumptions about how operators are executed. For example, Dedalus~\cite{dedalus} processes batches of data with a single global loop, while Naiad~\cite{naiad} processes messages one-by-one. In Flo, we generalize the requirement of deterministic processing into \textbf{eager execution}. This property enforces that Flo can \textbf{eagerly} execute downstream operators while their inputs are \textbf{still being updated}. Because we define this property in a way that allows for arbitrary execution schedules while arriving at a deterministic result, this gives a low-level scheduler significant power for deciding when operators should be run.

Flo is a declarative dataflow language that takes inspiration from the iterative processing of actors~\cite{actors}, but uses an event loop that maintains \emph{several} independent input and output queues. Rather than process messages one by one, programs in Flo describe a dataflow that operates over concrete collections of data. In fact, these collections are \emph{finite}, unlike models of streams such as co-inductive lists. To implement streaming applications, these concrete inputs can be extended, and the execution of the Flo program can be safely \emph{resumed} over these new inputs.

Flo also supports \textbf{streams of streams}, which capture behavior such as batching. Inspired by ingress/egress nodes in Naiad~\cite{naiad}, nested streams can be processed by \textbf{nested dataflow graphs}, which iteratively process chunks of data sourced from a larger stream with support for carrying state across iterations. This makes it possible to precisely implement a wide range of applications, such as windowed aggregations, processing data with minibatches, or iterative algorithms.

Flo is a \textbf{parameterized} family of languages which bring their own underlying data types and operators. Our proofs of streaming progress and eager execution are compositional, reducing the proof burden to individual operators. This allows Flo to capture the essence of a wide range of streaming systems under a single model, even allowing for composition that spans these approaches. To demonstrate this generality, we show how Flo can be used to model key ideas from a representative variety of streaming languages and incremental computation systems---Flink~\cite{flink}, LVars~\cite{lvars}, and DBSP~\cite{dbsp}---and show how existing semantic goals from each map to streaming progress and eager execution.

In summary, we make the following contributions:
\begin{itemize}
\item We formally define \textbf{streaming progress} and \textbf{eager execution} in the context of Flo, and specify a type system that reasons about stream termination (Section~\ref{sec:core}).
\item We introduce constructs in Flo for \textbf{composing operators} into dataflow graphs and prove that they preserve our key properties (Section~\ref{sec:composition}).
\item We describe the semantics of \textbf{nested streams and graphs} in Flo and demonstrate how they integrate with streaming progress and eager execution (Section~\ref{sec:nested}).
\item We show how the essence and key capabilities of \textbf{existing streaming languages} map to Flo and its foundational properties (Section~\ref{sec:case_studies}).
\end{itemize}

\section{Motivating Example}
\label{sec:motivation}

To understand why we need a model for streaming systems with strong semantic guarantees, let us walk through the challenges a developer may face while writing a simple program that sums up a stream of numbers.

We will accept a sequence of numbers from a streaming source, sum them up, and emit the resulting sum as the single fixed value in the output stream of our program. Streaming sources and sinks are modeled as inputs and outputs to a dataflow graph, so we will not have explicit operators for those. Instead, we can focus on just the core computation of summing up the numbers. A naive attempt may use a \texttt{fold} operator, which accumulates a value over a stream of data. In Rust:

\begin{verbatim}
    output = input.fold(0, |acc, x| acc + x)
\end{verbatim}

This program is simple, but it has a critical flaw: the \texttt{fold} operator is defined over a \emph{fixed} input collection. Operationally this means it will continue processing without producing any output until the stream somehow explicitly terminates. This concern is not addressed in the specification. In a streaming system, this is a common mistake that can lead to programs that hang indefinitely while consuming resources.

We next envision a number of ways a programmer could recognize and address this issue by choosing alternative semantics for this program. We categorize them into strategies that motivate the key properties we aim to establish with Flo: \textbf{streaming progress} and \textbf{eager execution}.

\subsection{Checking Boundedness Constraints}
Our program above does work on a subset of input streams: those that are finitely \textbf{bounded}, i.e. where the ``last'' element of the input stream is guaranteed to arrive. Unfortunately, this program is not well-defined on unbounded streams since we may accumulate the aggregation forever. In our semantics, we will model this failure case as an operator that does not satisfy \textbf{streaming progress}.

To resolve this, we can imagine classifying input streams via a subtype that would capture the boundedness property. We could then declare that the semantics of the \texttt{fold} operator are defined (correct) on bounded input streams, but undefined (incorrect) on unbounded input streams. Boundedness annotations on streams and operators would allow us to statically analyze the program above as incorrect, and suggest a fix: find a way to ensure that \texttt{input} is bounded. %

But what if the programmer's intent was to handle an unbounded \texttt{input} stream? Two natural variations to this specification are possible, as we discuss next.

\subsection{Coercing to Bounded Streams}
Many streaming applications and languages address the mismatch between unbounded streams and operators that require boundedness by introducing constructs 
for computing over finite batches or ``windows'' of the input stream~\cite{flink}. Perhaps this is what our programmer intended: their use of \texttt{fold} was intended to be scoped to a finite substream of \texttt{input}.

To capture this idea, we can envision a program variant that uses a \texttt{batch} operator to emit a \textbf{stream of streams}, where each inner stream is a batch of the original input. There are many possible ``windowing'' semantics for such a \texttt{batch} operator, but let us assume that any such \texttt{batch} operator ensures that each inner stream is \textbf{bounded} by specification. In that case, it is correct to employ \texttt{fold} over the inner streams, even though the outer stream may be \textbf{unbounded}. We can specify how each inner stream is handled via a \texttt{nest} operator that allows us to define a nested dataflow graph to run for each of these inner streams:

\begin{verbatim}
    output = input.batch().nest(|inner| {
        inner.fold(0, |acc, x| acc + x)
    })
\end{verbatim}

The output of this program will be another stream of streams, where each inner stream is the (single) sum of a batch of the input stream. This avoids the semantic problem of our previous example: even if \texttt{input} is unbounded, each \texttt{inner} argument to \texttt{nest} is bounded, and hence can be passed into \texttt{fold}. Moreover, if  \texttt{input} \emph{is} bounded, this program can (with appropriate parameterization) produce the same result as our original program above, by treating the whole input as a single batch. Hence in some sense we have not drifted too far from what seems to have been the programmer's original intent.

\subsection{Embracing Streaming Operators}
An alternative ``fix'' to the initial program would be to replace the \texttt{fold} operator with a streaming variant like \texttt{scan} that emits the ``running'' sum:
\begin{verbatim}
    output = input.scan(0, |acc, x| acc + x)
\end{verbatim}

On the positive side, this program works on both unbounded and bounded input streams (and it will satisfy our formal definition for streaming progress).
However, it seems rather distant from our original program: in particular, there is no way to make it produce the same result as our original program if \texttt{input} is bounded.

Instead, we could imagine a streaming operator whose output is a singleton stream of one monotonically growing value. At each step, this aggregator computes an updated sum, but ignores the result if it is smaller than the previous aggregated result. We could then write a program consuming an unbounded input stream:
\begin{verbatim}
    output = input.sum_lattice()
\end{verbatim}
Once again, for a bounded input, this program will produce the same result as our original program. It is, however, a departure from traditional streaming systems: for an unbounded input, the output of the \texttt{sum\_lattice} operator ``grows'' in the domain of natural numbers rather than in a domain of collections.

To get back to the domain of collections, such a ``monotonic singleton'' stream can be passed into a monotone function that emits an event upon reaching a threshold:
\begin{verbatim}
    output = input.sum_lattice().event_when_above(100)
\end{verbatim}

This is a common pattern in monitoring systems, and is a simplified version of the approach taken by LVars~\cite{lvars}. Why does the threshold need to be monotone? This boils down to our second formal property: \textbf{eager execution}. This requires that the overall program yields deterministic results even if we \emph{eagerly} execute operators on partial inputs. If this threshold were not monotone, there could be non-determinism due to when the threshold is evaluated. But \textbf{eager execution} is a more general property than monotonicity; we will show that it is equally meaningful in contexts where there is no natural ordering of values, such as in incremental computations with retractions.

\subsection{Discussion}
We started with a program that is ill-specified over unbounded streams. We saw various ways to ``fix'' this problem, inspired by salient design points of different streaming languages. What is key is that although these techniques were motivated by ideas from different languages, they all serve to satisfy two general properties of programs written in Flo: \textbf{streaming progress} and \textbf{eager execution}. In the following sections, we will walk through the formal semantics of Flo and show how we can precisely define these properties while retaining the flexibility to implement a wide range of streaming semantics found in the literature and used in practice.

\section{Collections, Streams, Operators, and Core Properties}
\label{sec:core}

The Flo model is based on specifications of dataflow pipelines, where \textbf{collections} of data elements are transformed by \textbf{operators} such as \texttt{map}, \texttt{filter}, or \texttt{join}. This is inspired by existing systems such as Flink~\cite{flink}, but with a critical difference that Flo is \emph{parameterized} over collection types and operators. This enables us to reason about a wide range of streaming paradigms and capture the essence of languages like LVars~\cite{lvars}, Bloom~\cite{bloom}, and Temporel~\cite{temporel} under a single model.

In this section, we define a family of collection languages $L^C$, operator languages $L^O$, and specify the formal properties that these languages must satisfy. In \cref{sec:composition}, we will define a new family of languages $L^G$ which include mechanisms to compose operators into a dataflow graph. Finally, in \cref{sec:nested}, we will extend $L^G$ with built-in operators for executing nested graphs. Our goal is to prove eager execution and streaming progress for all these languages.

\subsection{The Flo Event Loop}
Before we can dive into the semantics of these languages, we need to first discuss how Flo programs are executed. Flo deviates from classic streaming models in that it uses an actor-inspired event loop where messages are received, processed, and outputs are emitted. This means that the Flo program itself is always executing over concrete, finite collections of data rather than abstract streams. We describe a lightweight pseudocode for the event loop of a Flo program in \cref{fig:event_loop}.

\begin{figure}[h]
\begin{minipage}{0.94\textwidth}
    \begin{algorithm}[H]
    $O \gets \text{tuple of empty collections for each output}$

    $G \gets \text{initial Flo program}$

    \Loop{
        $\Delta \gets \text{tuple of new data batches for each input}$

        $I \gets \text{inputs of } G$

        $G \gets G \text{ with inputs set to } I \concat \Delta$

        $G, O \gets G \text{ after running an arbitrary number of small-steps with initial output } O$

        $O \gets \text{ remaining data after sending arbitrary part of } O$
    }
    \end{algorithm}
\end{minipage}

\caption{The event loop used to execute Flo programs.}
\label{fig:event_loop}
\end{figure}

Whenever a batch of new data is received, we use a \textbf{concatenation} operator $\concat$ to add this to the existing inputs. In classical streaming systems, such as those proposed in Flink~\cite{flink} and Stream Types~\cite{streamtypes}, this corresponds to appending new elements to the end of the existing data. But in Flo, our formalization makes it possible for this concatenation operator to take many forms, including those that do not monotonically grow the collection.

The other key aspect to note is that we run an \emph{arbitrary} number of small-steps of the program $G$ in each iteration, rather than running it until there is nothing to be done. We also allow the event loop to arbitrarily choose which data is sent at the end of each iteration; the outputs need not be consumed according to concatenation order. Later in this section, we will introduce key properties that ensure that this loop will always make progress and yield deterministic results.

\subsection{Collection Values, Expressions, and Types}
\label{sec:collection_types}

Flo programs manipulate \emph{collections}, which are concrete, finite values used to capture inputs, outputs, and (in \cref{sec:composition}) intermediate states of the program. Collection values can be updated as new data arrives or as an operator consumes data, but the way a collection value changes over time \emph{does not need to follow a partial order}, making it possible for our semantics to capture applications such as incremental computation over relations.

We define a collection language $L^C = (C, \concat, E^C, T^C, \llbracket\rrbracket^C, \lfloor\rfloor^C, \text{type}^C, \mathit{fix})$ as a tuple of:
\begin{itemize}
\item C: the set of collection values, which are mathematical objects
\item $\concat: C \times C \rightarrow C$: a ``concatenation'' function on collections %
\item $E^C$: the set of collection expressions, which are syntactic objects
\item $T^C \subseteq \mathcal{P}(C)$: the set of collection types, which are sets of collection values
\item $\llbracket\rrbracket^C: E^C \rightarrow C$: a total denotational semantics that maps collection expressions to values
\item $\lfloor\rfloor^C: C ⇀ E^C$: a partial lowering function that maps collection values to expressions %
\item $\mathit{type}^C: E^C \rightarrow T^C$: a total typing function that maps collection expressions to types
\item $\mathit{fix}: C \rightarrow C$, a transformation from a value into an equivalent\footnote{The definition of equivalence is up to the collection (for example, concatenating a stream terminator or setting a maximum size), and determines the guarantees provided by streaming progress (\cref{def:operator_streaming_progress})} one that is $\mathit{fixed}$
\end{itemize}

We additionally define: $\mathit{fixed}(c) \triangleq \forall{c' \in C}.~c \concat c' = c$ and $\emptydelta \in C$ is identity on the RHS of $\concat$.

We constrain $L^C$ via the following well-formedness conditions:
$$\forall{e \in E^C}.~\llbracket e \rrbracket^C \in \mathit{type}^C(e) ∧ ⌊〚e〛ꟲ⌋ꟲ = e$$
$$\forall{c \in C}.~ \mathit{fixed}(\mathit{fix}(c)) \land c \concat \emptydelta = c$$
$$\forall{\tau \in T^C, c \in \tau, c', c'' \in C}.~c \concat c' = c'' \implies c'' \in \tau$$

The language of collections involves both mathematical and syntactic representations. Our definition of collections is centered around collection \textbf{values}, which are the underlying mathematical objects being manipulated. At the syntax level, we represent these with collection \textbf{expressions}, which can be lifted to values via a denotational semantics, and then lowered back down to syntax using the $\lfloor\rfloor^C$ function. We also define a typing function $\mathit{type}^C$ that maps collection expressions to types, which are simply sets of collection values.%

A key difference between the Flo model and other
streaming semantics~\cite{streamtypes} is that the
concatenation function \textbf{does not} need to follow a partial order over
collection types, or satisfy algebraic properties like commutativity or associativity. What \emph{does} interest us is the question of when the concatenation function reaches a fixpoint. The $\mathit{fixed}$ predicate identifies a collection value such that no more data can be added to it, which we will leverage to define streaming progress.

Collections can take on a
variety of forms. A common collection in streaming systems is
the \emph{ordered sequence}, which captures an ordered list of
elements. But collections could also be multi-sets---as in streaming extensions to SQL~\cite{sql_stream}---or sets, as in Dedalus~\cite{dedalus}---where order often does not affect semantics. A ``collection'' can even be a
single value where ``concatenating'' to the collection updates the value---as in our \texttt{lattice\_sum} result in Section~\ref{sec:motivation}. We will lay out detailed examples of concrete collection types in \cref{sec:case_studies}.

\subsection{Stream Types and Boundedness}
\label{sec:stream_types}

Collections describe the values that are being processed by operators, but our discussion so far has been more reminiscent of batch processing than streaming. Our unique interest in streaming is the evolution of collections over time. In our motivation, we identified two key aspects of a streaming program's behavior with respect to time: \textbf{eager execution} makes it possible to correctly process newly-arrived data on an input to get an updated output, and \textbf{streaming progress} ensures that the program will not unexpectedly block on a collection becoming fixed.

To formally define streaming progress later in this section, we need to add a layer on top of collection types, which we call \textbf{stream types}. In our model, the key property we care about is whether a collection value will eventually become \textbf{fixed} (using the definition from \cref{sec:collection_types}), or if it may never become that. To capture this, we use a \textbf{boundedness flag} inspired by work in databases~\cite{tucker2003exploiting}, which is either \textbf{B}ounded or \textbf{U}nbounded. We define a stream type as a pair of a collection type and a flag on the left of \cref{fig:stream_types}. We will see stream types in action in \cref{sec:streaming_progress}.

\begin{figure}[h]
\centering
\vspace{0.4em}
\begin{minipage}{0.4\textwidth}
\begin{grammar}
<stream-type> ::= (<T>, $\mathbf{B}$ | $\mathbf{U}$)
\end{grammar}
\end{minipage}
\begin{minipage}{0.45\textwidth}
\begin{mathpar}\small
{
    \inferrule[reflexive-subtype]
    {}
    {S \leq S}
}

{
    \inferrule[bound-subtype]
    {}
    {(C, \mathbf{B}) \leq (C, \mathbf{U})}
}
\end{mathpar}
\end{minipage}

\caption{The grammar for stream types, where $T \in T^C$, and the subtyping relationship for stream types.}
\label{fig:stream_types}
\end{figure}

Note that collection expressions are not typed directly to a stream type, instead stream types are used as markers on inputs and outputs of a Flo program. We also have a simple subtyping relationship, where a stream type that is declared as bounded can be used in an unbounded context, because an unbounded stream has no restrictions on how the collection value behaves over time. We list the typing rule for this relationship on the right of \cref{fig:stream_types}, where $\leq$ is a subtyping relationship we will use in the rules for composing operators.

\subsection{Operators}
Flo programs transform input collections into output collections. This transformation is carried out by \textbf{operators} that consume data from several input collections to update output collections. In this section, we lay out the family of operator languages $L^O$, which captures Flo programs with a single operator. Because programs written in this language fit the general structure of the Flo event loop, we will use this language to lay out all the key properties we aim to prove about Flo. In \cref{sec:composition}, we will extend this language to $L^G$ to capture the composition of operators into a dataflow graph.

We will use the notation $[C]$ to represent tuples whose elements are each in C (and similarly for $[E^C]$), which denotes having multiple inputs or outputs. We will also denote $T^S$ to be the set of all stream types and $[T^S]$ to be a tuple of many stream types. Tuples of stream types follow an element-wise subtyping relationship.

We define an operator language $L^O = (L^C, E^O, \rightarrow^\delta, ORD^O, \vdash^O)$ as a tuple of:
\begin{itemize}
\item $L^C = (C, \concat, E^C, T^C, \llbracket\rrbracket^C, \lfloor\rfloor^C, \text{type}^C, \mathit{fix})$: a well-formed collection language
\item $E^O$: a language of operator expressions, which are syntactic objects
\item $(I, e^o) \rightarrow^\delta (I, e^o, O)$, a small-step operational semantics where $I, O \in [C]$ and $e^o \in E^O$
\item $(I, e^o) \prec^O (I, e^o) \in ORD^O$, a set of partial orders on collections where $I \in [C]$ and $e^o \in E^O$ (for some operators, we will omit the operator expression in the partial order)
\item $\vdash^O: e^O : (\tau^S \hookrightarrow \tau^S, \prec^O)$ a typing relation between elements $e^O \in E^O$, stream types $\tau^S \in [T^S]$, and partial orders $\prec^O \in ORD^O$
\end{itemize}

We augment this with the following definitions:
Given $L^O = (L^C, E^O, \rightarrow^O, ORD^O, \vdash^O)$, we define:
\begin{itemize}
\item The set of operator types: $T^O = \{ \tau_i \hookrightarrow \tau_o, \prec | \tau_i, \tau_o \in [T^S]~\land \prec \in ORD^O \}$
\item The small-step relation $\rightarrow^O = \{ ((I, e, O), (I', e', O \concat O'))~|~(I, e) \rightarrow^\delta (I', e', O') \}$
\item The typing relation on small-step configurations:\\ \inferrule{\vdash^O e : ((\tau_i, B_i) \ldots \hookrightarrow (\tau_o \ldots, B_o), \prec^O) \\ I \in (\tau_i \times \ldots) \\ O \in (\tau_o \times \ldots)}{\vdash^\rightarrow: (I, e, O) : (\tau_i \ldots \hookrightarrow \tau_o \ldots, \prec^O)} 
\end{itemize}

We further constrain $L^O$ via the following well-formedness condition:
$$\forall~{\prec} \in ORD^O.~\prec\text{ is finite and downwards-closed}$$

We also require, $\forall{e, e' \in E^O, \tau \in T^O, I,I',O,O' \in [C]}.~ \vdash^\rightarrow (I, e, O): \tau \land (I, e, O) \rightarrow^O (I', e', O')$ (For all well-typed expressions which step):
\begin{itemize}
\item $\rightarrow^O$ must be confluent
\item $\vdash^\rightarrow (I', e', O') : \tau$ (type preservation)
\item $\tau = (\ldots, \prec) \implies (e', I') \prec (e, I)$ (steps reduce the operator or its inputs)
\end{itemize}

Let us break down the intuition behind these properties. Every operator has a type with several input stream types and output stream types. The semantics of each operator are defined by the small-step relation $\rightarrow^\delta$, where the input and operator expression (which may carry state) are used to produce an updated input, operator expression, and an output collection. The small-step relation $\rightarrow^O$ transforms this relation into a classic operational semantics form, where the output generated by $\rightarrow^\delta$ is concatenated to the existing output (this concatenated form will be key to \cref{def:operator_eager}).

A key property of operators is the confluence of $\rightarrow^O$. In Flo, we \textbf{do not} require there to be a unique small step that can be taken for a given input and operator expression. For example, when processing a set of values, an operator may choose to process them in any order. But confluence guarantees that there exists some later state $(I', e', O')$ which all traces of small steps starting from $(I, e, O)$ will eventually reach. For operators that do have this non-determinism, proofs of this property typically involve a commutativity argument over the order of processing inputs.

Each operator also has a partial order over the operator expression and its inputs $\prec$, which is provided by the typing relation $\vdash^O$ and must be preserved across small-steps. We can use this to prove our first property on operators in $L^O$, that they always reach a stuck state in finite steps:

\begin{lemma}[Operator Stuck State]
\label{lem:operator_stuck}
Given an operator $op$, for all input states $I$ and output states $O$, there is a finite number of small steps that can be taken before no more small steps can be applied.
\end{lemma}

\begin{proof}
We leverage the partial order for this operator $\prec$. Since there are a finite number of operator expressions and collection values smaller than the initial state, and each step reduces the expression or its input, and the order is preserved across steps, there must be a finite number of total steps that can be taken before either no step applies or there is no smaller operator or input in the partial order.
\end{proof}

Note that our definition for stuck state does not require the expression to be reduced to some terminating form, such as the inputs all being empty. We only require that no more steps can be taken, which allows us to further loosen the requirements for collections; there is no need to define a unique bottom value, for example. Combined with the confluence of small-steps, this implies that every operator will eventually reach a unique stuck state.

\subsection{Eager Execution}
Flo hinges on two key properties that enable safe and progressive execution over streaming inputs: \textbf{eager execution} and \textbf{streaming progress}. The first guarantees that if new data arrives \emph{after} partial inputs have already been processed, then we can safely \emph{resume} the execution of the Flo program while arriving at a deterministic result. The second guarantees the program will never block on the fixedness of an input that may never become fixed. In \cref{sec:composition}, we will prove that both of these properties are true of \textbf{well-typed} graphs and Flo as a whole.

Eager execution avoids the situation where all input to an operator must be computed before the operator can begin execution. Instead, we require all operators to prove that they can begin processing partial inputs and receive additional data later via concatenation, while still producing the same result as if all the data was present from the start. This enables flexibility for scheduling and ensures that the outputs of a Flo program are deterministic even if an arbitrary number of small steps are run during each iteration of the event loop.

\newcommand{\cfg}{\ensuremath{\mathit{cfg}}}

\begin{definition}[Eager Execution]
\label{def:operator_eager}
Consider an operator $op \in E^O$. For all inputs $I \in [C]$, outputs $O \in [C]$, concatenated collection $\Delta \in [C]$, updated operator $op' \in E^O$, input collection, $I' \in [C]$, output collection $O' \in [C]$ such that $$(I, op, O) \rightarrow^O (I', op', O')\text{ and }(I \concat \Delta, op, O) \rightarrow^O (I'', op'', O'')$$

there exists a stuck state $(I''', op''', O''')$ such that

$$(I' \concat \Delta, op', O') {\rightarrow^{O*}}~(I''', op''', O''')$$

and $$(I'', op'', O'') \rightarrow^{O*}~(I''', op''', O''')$$
\end{definition}

Note that a simple inductive extension of this property tells us that we can introduce a single additional chunk of data of any size interleaved with executing small steps for the operator, and still end up in the same stuck state as if the data was present from the start. 
A further inductive argument says that if we have several chunks to concatenate, they can be introduced at any time interleaved with steps of the operator while still arriving at the same stuck state.

\subsection{Streaming Progress}
\label{sec:streaming_progress}

Streaming progress is a more challenging property to define. Unlike classic correctness properties such as determinism, streaming progress is focused on ensuring that outputs are kept \emph{fresh} with respect to certain inputs. Let us first formally define \emph{freshness} as \textbf{output maximality}.

\begin{definition}[Output Maximality]
We are given a well-typed (according to $\vdash^\rightarrow$) small-step configuration $((i_0 \ldots i_n), op, O)$ and well-typed final outputs $o_0' \ldots o_m'$ such that:

$((i_0, \ldots, i_n), op, O) \rightarrow^{O*} (I', op', (o_0', \ldots, o_m'))$ and $(I', op', (o_0', \ldots, o_m'))$ is stuck.

Then the given output $o_0' \ldots o_m'$ is \textbf{maximal} if $$((\mathit{fix}(i_0), \ldots, \mathit{fix}(i_n)), op, O) \rightarrow^{O*} ((i_0'', \ldots, i_n''), op'', (\mathit{fix}(o_0'), \ldots, \mathit{fix}(o_m')))$$

and $((i_0'', \ldots, i_n''), op'', (\mathit{fix}(o_0'), \ldots, \mathit{fix}(o_m')))$ is stuck.
    
\end{definition}

Consider our motivating example. Some operators (\texttt{scan}) can satisfy Output Maximality for all inputs because at any point in the execution, we can reach a state where all outputs are released, and no more outputs would be released if the input became fixed. But other operators (\texttt{fold}) cannot satisfy Output Maximality for all inputs, because we never reach a state with any outputs released unless the input is fixed, at which point the output is released (and hence changes).

This is where the stream types we introduced earlier come in, which will allow us to define a property for streaming progress that works for all operators. 
Each operator annotates its inputs and outputs with boundedness flags. Intuitively, if an \textbf{input} is \textbf{unbounded}, we want to prevent the problem we have illustrated with \texttt{fold}: we do not want the operator to block until the input becomes fixed.
By contrast, if an input is \textbf{bounded}, it may make sense for an operator (e.g., \texttt{fold}) to withhold some outputs until the input becomes fixed. 

Output Maximality and stream types together enable us to
ensure that an operator always keeps its outputs as \emph{fresh} as possible: bounded inputs are guaranteed to produce outputs (after becoming fixed), as are unbounded inputs (since they do not block on fixedness).

Finally, to enable composition across multiple operators, we want to derive restrictions on the outputs from input properties.
Once the \textbf{bounded inputs} are fixed, the \textbf{bounded outputs} must become fixed in a finite number of steps to avoid blocking downstream operators. With that intuition in place, we formally define streaming progress in terms of Output Maximality:

\begin{definition}[Streaming Progress]
\label{def:operator_streaming_progress}

Consider a well-typed operator $op$ with type $\vdash^O op: ((I_0, B_{I,0}) \ldots (I_n, B_{I,n})) \hookrightarrow ((O_0, B_{O,0}) \ldots (O_m, B_{O,m}))$. Consider all well-typed inputs $i_0 \ldots i_n \in C$ such that $B_{I,j} = \mathbf{B} \implies \mathit{fixed}(i_j)$ (the bounded \textbf{inputs} are fixed).

Let us also consider all well-typed initial outputs $O$ and final outputs $o_0' \ldots o_m'$, such that:
$$((i_0, \ldots, i_n), op, O) \rightarrow^{O*} (I', op', (o_0', \ldots, o_m'))$$

and $(I', op', (o_0', \ldots, o_m'))$ is stuck. Then the operator $op$ satisfies streaming progress if:
\begin{itemize}
\item $o_0' \ldots o_m'$ are \textbf{maximal} for the operator $op$ with inputs $i_0 \ldots i_n$ and initial outputs $O$
\item $\forall{j}.~B_{O,j} = \mathbf{B} \implies \mathit{fixed}(o_j')$ (the bounded \textbf{outputs} are fixed)
\end{itemize}
\end{definition}

Any operator in an implementation of Flo must satisfy these properties. We will show in the next section that these properties are automatically preserved when composing operators into graphs, which alleviates any further proof burden for the implementation.

\section{Composition with Graphs}
\label{sec:composition}
Programs in Flo are formed by composing operators into a directed-acyclic graph, where each node is an operator and each edge captures an intermediate collection of data elements. In Flo, we express these directed acyclic graphs as expressions of $L^G$ through recursive constructs for sequential and parallel composition, such as in \cref{fig:dataflow_graph}.

\begin{figure}[h]
\begin{minipage}{0.45\textwidth}
\includegraphics[page=1,trim={14cm 12.5cm 18cm 8cm},clip,width=\textwidth]{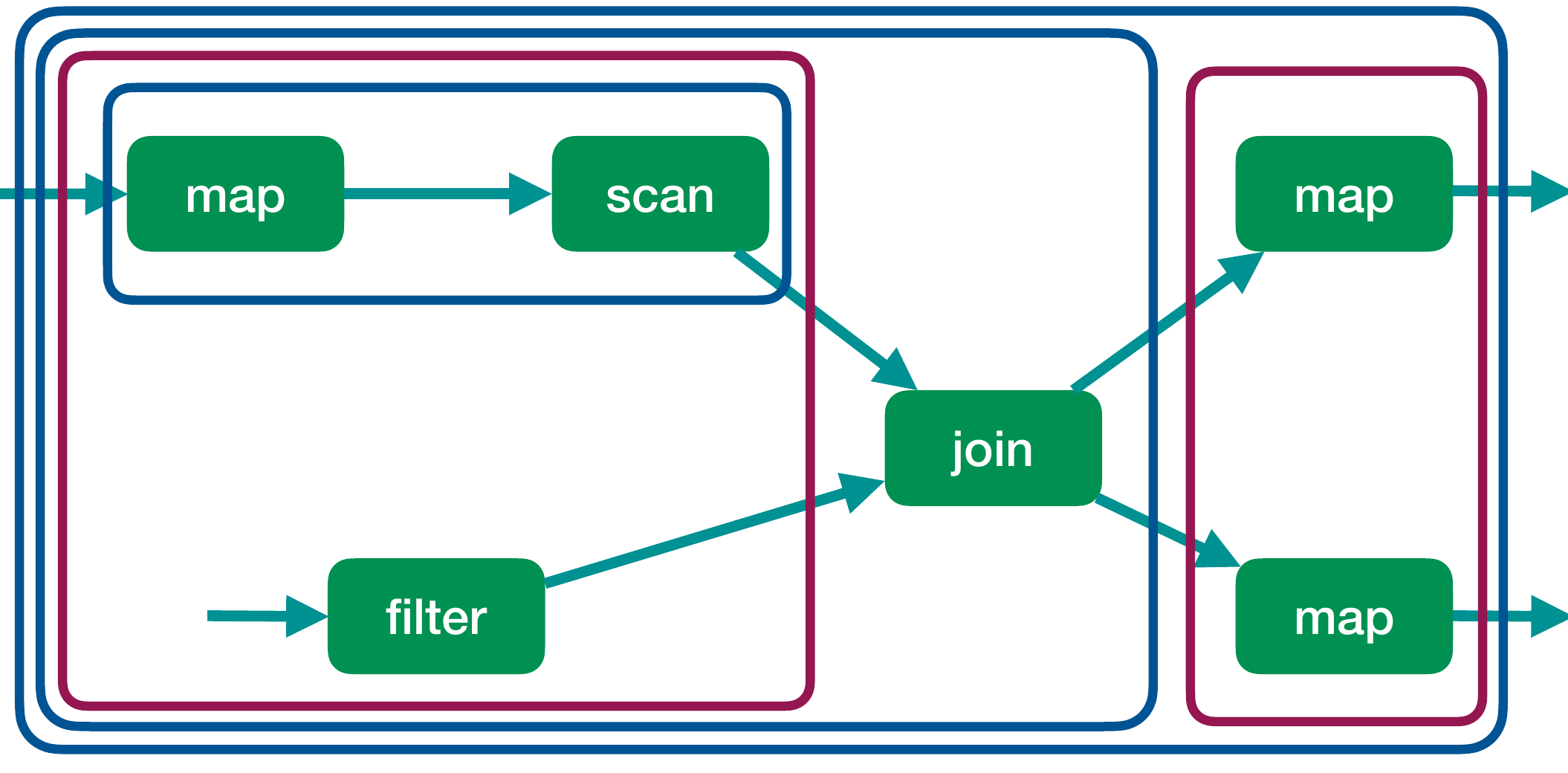}
\end{minipage}
\hspace{0.03\textwidth}
\begin{minipage}{0.35\textwidth}
\begin{verbatim}
(
  (
    (({_} map); ({_} scan)) |
    ({_} filter));
  ({_} join));
(({_} map) | ({_} map))
\end{verbatim}
\end{minipage}

\caption{A dataflow graph and its decomposition into an expression in our language (with parentheses for clarity). Magenta boxes represent parallel composition and blue boxes represent sequential composition.}
\label{fig:dataflow_graph}
\end{figure}

Unlike before, the graph language $L^G$ is not parameterized on any new definitions, and can be directly layered on any instance of an operator language $L^O = (L^C, E^O, \rightarrow^\delta, ORD^O, \vdash^O)$. We layer on this language a few additional constructs:
\begin{itemize}
\item $E^G$: the language of graph expressions, which are syntactic objects (\cref{fig:graph_grammar})
\item $\vdash: e^G : (\tau^S \hookrightarrow \tau^S, \prec)$ a typing relation between elements $e^G \in E^G$, stream types $\tau^S \in [T^S]$, and partial orders $\prec \in \bigcup_{n \in \mathbb{N}} (ORD^O)^n$ (\cref{fig:graph_types})
\item $e^g \rightarrow^\Delta (e^g, O)$, a small-step operational semantics where $O \in [C]$ and $e^g \in E^G$ (\cref{fig:graph_operational})
\end{itemize}

We will also augment this with the small-step relation: $\rightarrow~= \{ ((g, O), (g', O \concat O')) | g \rightarrow^\Delta (g', O') \}$

Sequential composition passes the outputs of one subgraph into the inputs of the other, and is the primary way that operators can be chained together in a Flo program. Parallel composition makes it possible to capture portions of the graph where several operators can be run independently on separate sets of inputs to produce separate outputs. We lay out the grammar for graphs in \cref{fig:graph_grammar}.

\newcommand{\es}{\ensuremath{e}}

\begin{figure}[h]
    \centering
    \begin{minipage}[t]{.9\textwidth}
    \begin{align*}
    \es ::=~& \es| \es~∣~\es ; \es~∣~\{S\}[O]
    \end{align*}
    \end{minipage}

    \caption{The grammar for graphs of a Flo program, where $S \in [E^C]$ and $O \in E^O$.}
    \label{fig:graph_grammar}
\end{figure}

Note that we include a state term $S$, which collects inputs to an operator.  This term will be essential when formalizing our small-step semantics, which needs to reason about buffered inputs at an {\it arbitrary} position in a graph. 
Our type system models graphs in terms of their input and output stream types, and a partial order over inputs like for operators. We list the typing rules for graphs in \cref{fig:graph_types} and small-step operational semantics in \cref{fig:graph_operational}. In our semantics, we will use $\cdot$ to denote tuple concatenation, when dealing with types or values.

\begin{figure}[h]
    \begin{mathpar}\small
    {
        \inferrule[sequence]
        { ⊢ \es₁ : (I_1 ↪ O_1, \prec_1) \\ ⊢ \es₂ : (I_2 ↪ O_2, \prec_2) \\\\ O_1 \leq I_2}
        { ⊢ \es₁;\es₂ : (I_1 ↪ O_2, \prec_1) }
    }

    {   
        \inferrule[par]
        {⊢ \es₁ : (I_1 ↪ O_1, \prec_1) \\ ⊢ \es₂ : (I_2 ↪ O_2, \prec_2)}
        { ⊢ \es₁ ∣ \es₂ : (I_1 \cdot I_2 ↪ O_1 \cdot O_2, \prec_1 \cdot \prec_2) }
    }

    {
        \inferrule[operator]
        {\vdash^O op : (I ↪  O, \prec)
        \\
        I = ((S_0, B_0), \ldots (S_n, B_n))
        \\ \forall{i}.~\mathit{type}^C(s_i) = S_{i}}
        { ⊢ \{(s_0, \ldots, s_n)\}[op] : (I ↪ O, (\prec))}
    }
    \end{mathpar}

  \caption{Type semantics for graphs of a Flo program.}
  \label{fig:graph_types}
\end{figure}

\newcommand{\inputs}{\ensuremath\textit{inputs}}
\newcommand{\newinputs}{\ensuremath\textit{setinput}}

\begin{figure}[h]
    \begin{align*}
        \inputs(\es₁;\es₂) ≜&~\inputs(\es₁)\\
        \inputs(\es₁∣\es₂) ≜&~\inputs(\es₁) \cdot \inputs(\es₂)\\
        \inputs(\{I\}[op]) ≜&~I
    \end{align*}

    \begin{align*}
        \newinputs(\es₁;\es₂,I) ≜&~\newinputs(\es₁,I);\es₂&\\
        \newinputs(\es₁∣\es₂,I_1 \cdot I_2) ≜&~\newinputs(\es₁,I_1)∣\newinputs(\es₂,I_2)&\\
        \newinputs(\{I\}[op],I') ≜&~\{I'\}[op]&\textit{when}~|I| = |I'|
    \end{align*}

    \begin{mathpar}\small
    
    {
        \inferrule[sequence-left]
        {\es₁ \rightarrow^\Delta ({\es}₁',I')}
        {(\es₁;\es₂) \rightarrow^\Delta ({\es₁}';\newinputs(\es₂, \lfloor\llbracket\inputs(\es₂)\rrbracket^C \concat I'\rfloor^C),\emptydelta)}
    }

    {
        \inferrule[sequence-right]
        {\es₂ \rightarrow^\Delta (\es₂',O')}
        {(\es₁;\es₂) \rightarrow^\Delta (\es₁;\es₂',O')}
    }

    {
        \inferrule[par-left]
        {\es₁ \rightarrow^\Delta (\es₁',O₁') }
        {(\es₁ ∣ \es₂) \rightarrow^\Delta (\es₁' ∣ \es₂, O₁',\emptydelta)}
    }

    {   
        \inferrule[par-right]
        {\es_2 \rightarrow^\Delta (\es_2',O_2') }
        {(\es₁ ∣ \es₂) \rightarrow^\Delta (\es₁ ∣ \es₂', \emptydelta,O₂')}
    }

    {
        \inferrule[operator]
        {(\llbracket{I}\rrbracket^C, op) \rightarrow^\delta (I',op',O')}
        {(\{I\}[op]) \rightarrow^\Delta (\{\lfloor{I'}\rfloor^C\}[op'],O')}
    }
    
    \end{mathpar}

  \caption{Small-step semantics for graphs of a Flo program.}
  \label{fig:graph_operational}
\end{figure}

Before we continue, let us prove that graphs satisfy preservation.

\begin{lemma}[Graph Preservation]
\label{lem:graph_preservation}

Given a graph $g$ of type $(I ↪ O, \prec)$, output state $S = (s_0 \ldots s_n)$, and updated output state $S' = (s_0' \ldots s_n')$ such that $O = ((T_0, \_) \ldots (T_n, \_))$ and $\forall{i}.~\mathit{type}^C(s_i) = T_i$, if $(g, S)$ takes a step to $(g', S')$, then $g'$ is also of type $(I ↪ O, \prec)$ and $\forall{i}.~\mathit{type}^C(s_i') = T_i$.
\end{lemma}

\begin{proof}
We can prove this by structural induction over the graph.

\textbf{Base Case:} A graph with a single operator. By operator preservation, we know that the type of $I$ is the same as the type of $I'$, that $op'$ has the same type, and that $O'$ has the same type as $O$. Therefore, the graph as a whole has the same type and the output is of the correct type.

\textbf{Inductive Step:} Proof by cases:

\textbf{Sequential Composition}: If we apply the sequence-left rule, then by induction we know that $e_1$ has the same type as $e_1'$, and $I'$ has the same types as the inputs of $e_2$. Therefore, when we set the inputs of $e_2$ to $I'$, we preserve the typing (due to well-formedness of the denotational lifting and syntactical lowering). Since the output is unchanged, we satisfy preservation.

If we apply the sequence-right rule, then by induction we know that $e_2$ has the same type as $e_2'$, and the output has the same type due to concatenation. Therefore, we satisfy preservation.

\textbf{Parallel Composition}: In both rules, we use induction to know the types of both sides are preserved. The typing rule for parallel simply composes these types, so we are done.
\end{proof}

\subsection{Graph Stuck State}
Now, let us extend the properties we require of operators to graphs as a whole. First, we will extend Operator Stuck State (\cref{lem:operator_stuck}).

\begin{lemma}[Graph Stuck State]
\label{lem:graph_stuck}
Given a graph initialized with a fixed set of input collection values, running the graph will eventually reach a stuck state where no additional steps can be taken.
\end{lemma}

\begin{proof}
We can prove this by structural induction over the graph.

\textbf{Base Case:} A graph with a single operator. By \cref{lem:operator_stuck}.

\textbf{Inductive Step:} A graph such that its subgraphs satisfy Graph Stuck State. Proof by cases:

\textbf{Sequential Composition}: There are only two small steps that can be taken at any point, for the left or right. If we only step one of the two subgraphs, by induction that side will eventually reach a stuck state. If the left side reaches a stuck state, then running the right side will never re-enable the left side by the definition of $\rightarrow$. If the right side reaches a stuck state, we may be able to run the left side which may re-enable the right side, but this will cycle back to the left and eventually the left side will be stuck. Therefore, the graph as a whole will reach a stuck state.

\textbf{Parallel Composition}: The two subgraphs are independent, and so by the inductive hypothesis we know that both will eventually reach a stuck state, and their composition is a stuck state.
\end{proof}

\subsection{Determinism and Eager Execution}
The most significant change between reasoning about operators in isolation and the composition of them is that at any point when executing a graph, there may be multiple small steps for each operator that can be taken. We need to prove we can non-deterministically execute these operators while arriving at the same output. To prove this for all graphs, we will also need to extend Eager Execution to graphs. These proofs are mutually recursive, so we will prove them simultaneously. Both our definitions look nearly identical to those for operators, just with the use of the general small step relation rather than only for operators.

A quick aside on notation. In this section, we will use the shorthand $\{ I \} g$ to denote a graph $g$ whose inputs are set to $I$, so $\{ I \} g = \mathit{setinput}(g, I)$.

\begin{definition}[Determinism]
\label{def:graph_determinism}
Consider a graph $g$. For all inputs $I$ and initial outputs $O$ where a small step for $(\{ I \} g, O)$ exists, there exists a later state $g'$, inputs $I'$, and outputs $O'$ such that in every trace of small steps $(\{I\} g, O) \rightarrow^* (\{I'\} g', O')$ we eventually reach this later state.
\end{definition}

Note that combined with stuck states (\cref{lem:graph_stuck}), this implies that every graph will eventually reach a \textbf{unique} stuck state. This is because we can always take a series of steps to arrive at the same later state, and eventually we will reach a point where no more steps can be taken.

\begin{definition}[Eager Execution]
\label{def:graph_eager}
Consider a graph $g$. For all input streams $I$, output streams $O$, delta set $\Delta$, updated graph $g'$, input stream, $I'$, and output stream $O'$ such that

$$(\{ I \} g, O) \rightarrow (\{ I' \} g', O')$$ there exists a stuck state $f$ such that

$$(\{ I \concat \Delta \} g, O) {\rightarrow^{*}}~f\text{ and }(\{ I' \concat \Delta \} g', O') \rightarrow^{*}~f$$
\end{definition}

\begin{lemma}
\label{lem:graph_determinism_eager}
Consider any expression. It must satisfy:
\begin{enumerate}
    \item Determinism
    \item Eager Execution
\end{enumerate}
\end{lemma}

\begin{proof}
We can prove this by structural induction over the graph.

\textbf{Base Case:} A graph with a single operator.
\begin{enumerate}
\item By confluence of $\rightarrow^O$.
\item By \cref{def:operator_eager}.
\end{enumerate}

\textbf{Inductive Step:} A graph such that its subgraphs satisfy both (1) and (2). Proof by cases:

\textbf{Sequential Composition}: a graph of form $a ; b$
\begin{enumerate}
\item We know that there is at least one small-step that can be taken, and the only options are to recursively step $a$ or $b$. Let us define an \emph{execution trace} that captures an ordered sequence of small-steps to take. This trace will have the form ``$(a_i | b)+$'', with each element directing us to take the corresponding small step corresponding to the named subgraph, with the indices for $a$ counting up from $0$. Given a trace $t = $``$s_0 \ldots s_n$'', we define $\rightarrow_t$ to take the steps in order. For each instance of $a_i$, the index lets us uniquely identify the small-step rule that will be applied to $a$. For $b$, the token represents taking any small-step on $b$. We will call a trace after which no more steps can be taken a \emph{terminating trace}.

Next, let us define equivalence between a pair of traces $t_1$ and $t_2$. Two traces are equivalent if executing both on the same initial state results in the same final state, \emph{even} with non-deterministic selection of which small-step to run for each $b$. We will prove that for any pair of terminating traces $t_1$ and $t_2$, the traces are equivalent.

Consider a trace of the form ``$\mathit{prefix}~b~a_i~\ldots~a_j~b^*$''. The execution of this looks like $$(\{I_a^p\} a^p ; \{I_b^p\} b^p, O^p) \rightarrow_{\mathit{prefix}} (\{I_a\} a ; \{I_b\} b, O) \rightarrow_{b} (\{I_a\} a ; \{I_b'\} b', O')$$
$$\rightarrow_{a_{i \ldots j}} (\{I_a'\} a' ; \{I_b''\} b', O') \rightarrow_{b}^{*} (\{I_a'\} a' ; \{I_b'''\} b'', O'')$$

First, by the definition of $\rightarrow^\Delta$, we know that $I_b'' = I_b' \concat \Delta_i \concat \Delta_{i + 1} \ldots$. Then, inductively Eager Execution applied to $b$ lets us rewrite ``$b~a_i~\ldots~a_j~b^*$'' to ``$a_i~\ldots~a_j~b^*$'' (note that the number of trailing $b$ in the rewritten suffix may be arbitrary), because the execution of $a_i~\ldots~a_j$ simply introduces additional data for $b$ to process. This results in the following execution $$(\{I_a^p\} a^p ; \{I_b^p\} b^p, O^p) \rightarrow_{\mathit{prefix}} (\{I_a\} a ; \{I_b\} b, O)$$
$$\rightarrow_{a_{i \ldots j}} (\{I_a'\} a' ; \{I_b \concat \Delta_i \concat \Delta_{i + 1} \ldots\} b, O) \rightarrow_{b}^{*} (\{I_a'\} a' ; \{I_b'''\} b'', O'')$$

Therefore, the trace $\mathit{prefix}~b~a_i~\ldots~a_j~b^*$ is equivalent to $\mathit{prefix}~a_i~\ldots~a_j~b^*$.

If we repeatedly apply this rewrite to both traces to pull all $a_i$ to the front, we will arrive at two traces of the form $a_0~\ldots~a_n~b^*$ and $a_0~\ldots~a_m~b^*$. We know that both original traces are terminating, therefore after running $a_0~\ldots~a_n$ and $a_0~\ldots~a_m$ even though the $b$s between the elements have been removed, there will be no more small steps that can be taken on $a$. By determinism from induction, since $a$ has terminated the traces $a_0~\ldots~a_n$ and $a_0~\ldots~a_m$ result in the same state and are equivalent. Similarly, because our rewrites preserve equivalence, by determinism we know that after running $b^*$ on both traces, we will reach the same final state. Therefore, the traces are equivalent and $a ; b$ satisfies determinism.

\item We can split into cases based on the small step that could be taken.

\textbf{Case 1:} The small step is on $a$. By \cref{def:graph_eager}, we know that we can introduce the delta before or after the small step on $a$ and then continue running small steps for $a$ until reaching the common later state for $a$, which is also our overall later state $f$.

\textbf{Case 2:} The small step is on $b$. If we run the small step, then introduce the delta, let the state immediately after introducing the delta be $f$. If we instead first introduce the delta, then run $b$, the state after is also $f$ because running the small step for $b$ is unaffected by the introduction of the delta.
\end{enumerate}

\textbf{Parallel Composition}: a graph of form $a | b$
\begin{enumerate}
\item The small steps for a parallel composition just run the small steps for either side, which are independent. Therefore by induction both sides will step to a deterministic state.
\item In parallel composition, the introduction of a delta results in independent chunks being added to both sides. If we step the graph first, that just steps one of the sides, so the inductive hypothesis holds on one of the sides and the other side is unaffected.
\end{enumerate}
\end{proof}

\subsection{Streaming Progress}
\begin{lemma}[Streaming Progress for Graphs]
\label{lem:graph_streaming_progress}

Consider a well-typed graph $g$ with type $\vdash g: (((I_0, B_{I,0}) \ldots (I_n, B_{I,n})) \hookrightarrow ((O_0, B_{O,0}) \ldots (O_m, B_{O,m})), \prec)$ such that $\mathit{inputs}(g) = i_0 \ldots i_n$ and $B_{I,j} = \mathbf{B} \implies \mathit{fixed}(i_j)$. Consider all well-typed outputs $O$ and $o_0' \ldots o_m'$ such that $$(g, O) \rightarrow^* (g', (o_0', \ldots, o_m'))$$ and $(g', (o_0', \ldots, o_m'))$ is stuck. Then $o_j'$ must be fixed if $B_{O,j} = \mathbf{B}$ and there must also be a stuck state $$(\{(\mathit{fix}(i_0), \ldots, \mathit{fix}(i_n))\} g, O) \rightarrow^* (g'', (\mathit{fix}(o_0'), \ldots, \mathit{fix}(o_m')))$$    
\end{lemma}

\begin{proof}
We can prove this by structural induction over the graph.

\textbf{Base Case:} A graph with a single operator. By \cref{def:operator_streaming_progress}.

\textbf{Inductive Step:} Proof by cases:

\textbf{Sequential Composition}: We can apply \cref{lem:graph_determinism_eager} to only focus on traces where we run the left half until stuck state and then the right half. First, we apply streaming progress to the left half, which tells us that we will output intermediate collections such that each output with a bounded stream type will have a fixed value. This satisfies the premise for induction on the right subgraph, so we can apply streaming progress again to know that each bounded output will be fixed. Using the same proof structure, we know that the intermediate collections will be maximal with respect to the unbounded inputs, so the final outputs will be maximal as well.

\textbf{Parallel Composition}: Because both sides are independent, we can simply use induction on each side. Because all bounded outputs will be fixed and all outputs are maximal with respect to the unbounded inputs, we satisfy streaming progress.
\end{proof}

\section{Nested Streams and Graphs}
\label{sec:nested}

So far, we have considered dataflow programs with a direct path of operators from each input to the outputs. But for many applications, it is necessary to perform stateful, iterative computations over an input stream. In Flo, we tackle this using constructs for \textbf{nested streams and graphs}.

Before we dive into formal semantics, let us lay out a high-level overview of our approach to nesting. First, we introduce nested streams, which are a specific type of stream that encapsulate several smaller streams. We define a set of restrictions for how operators must generate such nested streams, in particular how boundedness of the inner streams is enforced.

Once we have nested streams, we need an operator that can transform them. This is where the \texttt{nest} operator comes in, which makes it possible to transform a nested stream by defining a nested Flo graph that should be run on each inner stream. We introduce the \texttt{write\_defer} and \texttt{read\_defer} operators, which can be used to pass state across the iterations for each inner stream to enable iterative computation. We prove that these operators satisfy all the core operator properties, therefore preserving the high-level guarantees we have established for Flo.

\subsection{Nested Streams}
Our definition of Flo so far has dealt only with an abstract notion of collections and operators. But the \texttt{nest} operator is a concrete instance, and so we also need a concrete collection type for it to consume and produce. Furthermore, this collection type must store nested streams in a way that preserves boundedness properties and allows the inner graph to manipulate the inner streams.

To tackle this, we introduce the \textbf{ordered sequence of streams} in \cref{fig:ordered_sequence_of_streams}. This collection type, denoted $[(S_0, \ldots S_n)]$ is parameterized over several inner stream types $S_i = (C_i, B_i)$. Values of this type are stored as a list of tuples $[(c_{0, 0}, \ldots c_{0, n}), \ldots, (c_{m, 0}, \ldots c_{m, n})]$, where each $c_{i, j}$ is a value of type $C_j$. The terminator symbol $\terminator$ indicates the end of a stream.

\begin{figure}[h]
\begin{align}
[((C_1, &B_1), \ldots, (C_m, B_m))] \triangleq \{~[(c_{1,1},~\ldots), \ldots, (c_{n,1},~\ldots)]~| \nonumber\\
    &\forall{i,j}~c_{i,j} \in C_j \land (i > 1 \land B_j = \mathbf{B}) \implies \mathit{fixed}(c_{i,j}) \nonumber\\
\}~\cup~\{&~[\terminator, (c_{1,1},~\ldots), \ldots, (c_{n,1},~\ldots)]~| \nonumber\\
    &\forall{i,j}~ c_{i,j} \in C_j \land (B_j = \mathbf{B}) \implies \mathit{fixed}(c_{i,j}) \} \nonumber
\end{align}
$$[\terminator, \ldots] \concat x = [\terminator, \ldots]$$
$$[c_1, \ldots, c_n] \concat \terminator =[\terminator, c_1, \ldots, c_n]$$
$$[c_1, \ldots, c_n] \concat ((v_1, \ldots, v_m), \mathit{true}) =[(v_1, \ldots, v_n), c_1, \ldots, c_n]$$
$$[(v_1, \ldots, v_m), \ldots, c_n] \concat ((\delta_1, \ldots, \delta_m), \mathit{false}) = [(v_1 \concat \delta_1, \ldots, v_m \concat \delta_m), \ldots, c_n]$$

\caption{The collection type and concatenation operator for the ordered sequence of streams.}
\label{fig:ordered_sequence_of_streams}
\end{figure}

The concatenation operator on this collection type takes an ordered sequence of streams and \emph{either} the terminator $\terminator$, the tuple of the boolean true and a tuple of collections values matching the inner stream types, or a tuple of the boolean false and a tuple of concatenation values corresponding to the right-hand side accepted by $\concat$ for each inner stream type. If the boolean flag is true, the concatenation operator extends the collection with the tuple used as the new leftmost value. If it is false, the operator uses the concatenation operator of each of the inner stream types to extend the existing leftmost collections with the new values.

There is another key concern we need to address. Once a new tuple of collections is pushed into the ordered sequence, none of the other tuples will ever grow through concatenation. We need to ensure that these finalized tuples satisfy the restrictions of the inner stream types; in particular that they satisfy boundedness properties. To do this, we require that all tuples of collections after the leftmost one have fixed collections for each bounded stream type.

\subsection{Nesting Graphs}
The \texttt{nest} operator maps nested streams by transforming their inner streams one-by-one using an inner Flo graph. These inner graphs have special privileges: they can define \emph{iterative computations} by passing data across executions on subsequent inner streams. To do this, developers use pairs of \texttt{read\_defer} and \texttt{write\_defer} operators with matching keys. Any data sent to a \texttt{write\_defer} operator will be emitted by the corresponding \texttt{read\_defer} operator when processing the next inner stream (for the first step, \texttt{read\_defer} takes an initial value as a parameter).%

\diffblock{
Before we dive into the formal semantics of these operators, let us walk through a simple example to show how \texttt{nest}, \texttt{write\_defer}, and \texttt{read\_defer} can be combined to enable iterative computation. We will implement a classic iterative algorithm where we are given a set of directed edges and want to compute which nodes are reachable from a root within a fixed radius. Our algorithm starts with a single root node, and in a loop identifies the next ``layer'' of reachable nodes.}
\begin{DIFnomarkup}

\begin{figure}[h]
    \centering
    \includegraphics[page=3,trim={0 10cm 0 6cm},clip,width=0.8\textwidth]{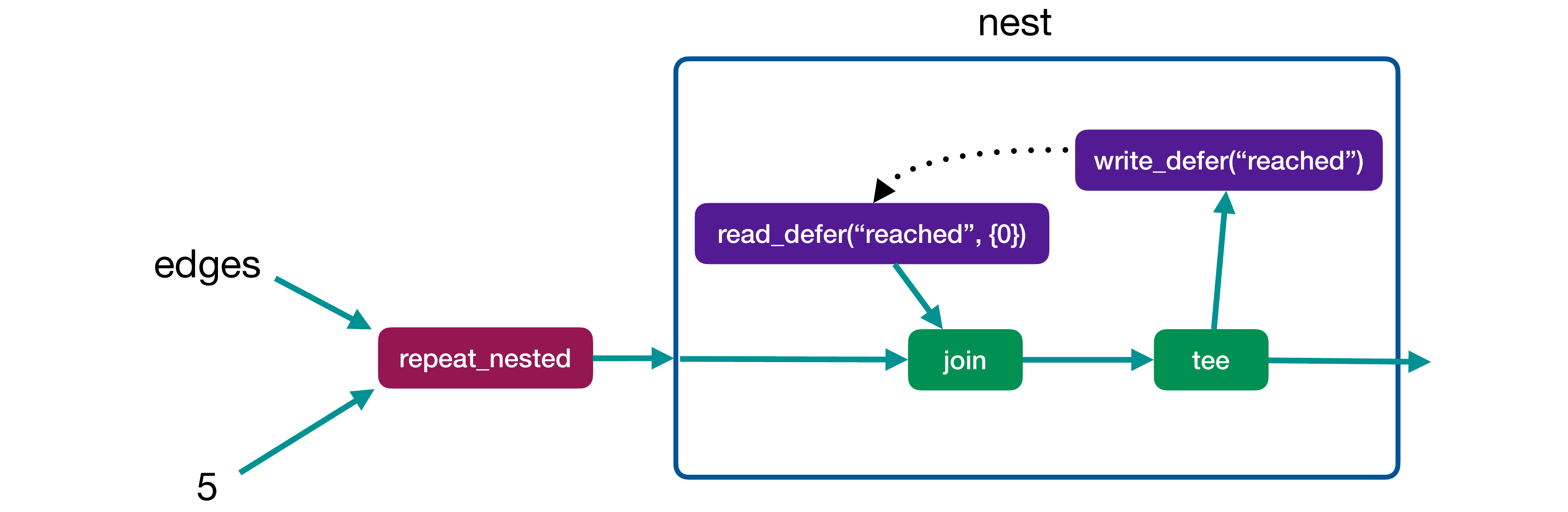}
    \caption{An example of identifying nodes within a fixed radius using nested graphs.}
    \label{fig:graph_reachability_example}
\end{figure}
\end{DIFnomarkup}

First, we need a collection type for sets of nodes and sets of edges (using standard semantics), along with some operators inspired by relational algebra. We omit the detailed semantics for brevity, but these are straightforward to define. The \texttt{join} operator takes in a set of nodes and a set of edges, and identifies the destination of all edges originating at a node in the input set. The \texttt{tee} operator consumes a single stream and emits a pair of streams, each duplicating the input.

Next, we must generate a stream-of-streams that drives the nested graph. For graph reachability within a fixed radius $n$, we need to run $n$ iterations of the inner graph. To achieve this, we introduce a \texttt{repeat\_nested} operator which consumes a stream and a natural number singleton $k$, and emits a stream with $k$ inner streams, each of which duplicates the contents of the input.

Putting these operators together, we show how to implement this algorithm in \cref{fig:graph_reachability_example}. On every iteration, we first collect the nodes reached up to the previous iteration using \texttt{read\_defer}, with an initial value of just the root node $0$. Then, we emit the next layer of reachable nodes and also send them to \texttt{write\_defer} to be used in the next iteration. In the output of this program, we will have a stream of sets of nodes, where each set contains the nodes reachable from the root with increasing radii up to the fixed limit.

In Flo, \texttt{nest} is a standard operator that satisfies all the proof obligations, so it can be... nested! This makes it possible to build arbitrarily complex nested cycles. For example, we can tweak the graph reachability example to allow recomputing the reachability analysis with extended radii. In this algorithm, we can use the output from a previous query to ``bootstrap'' the next query, and only run iterations to extend the radius rather than starting from scratch.

\begin{figure}[h]
    \centering
    \includegraphics[page=4,trim={0 3cm 0 5cm},clip,width=0.85\textwidth]{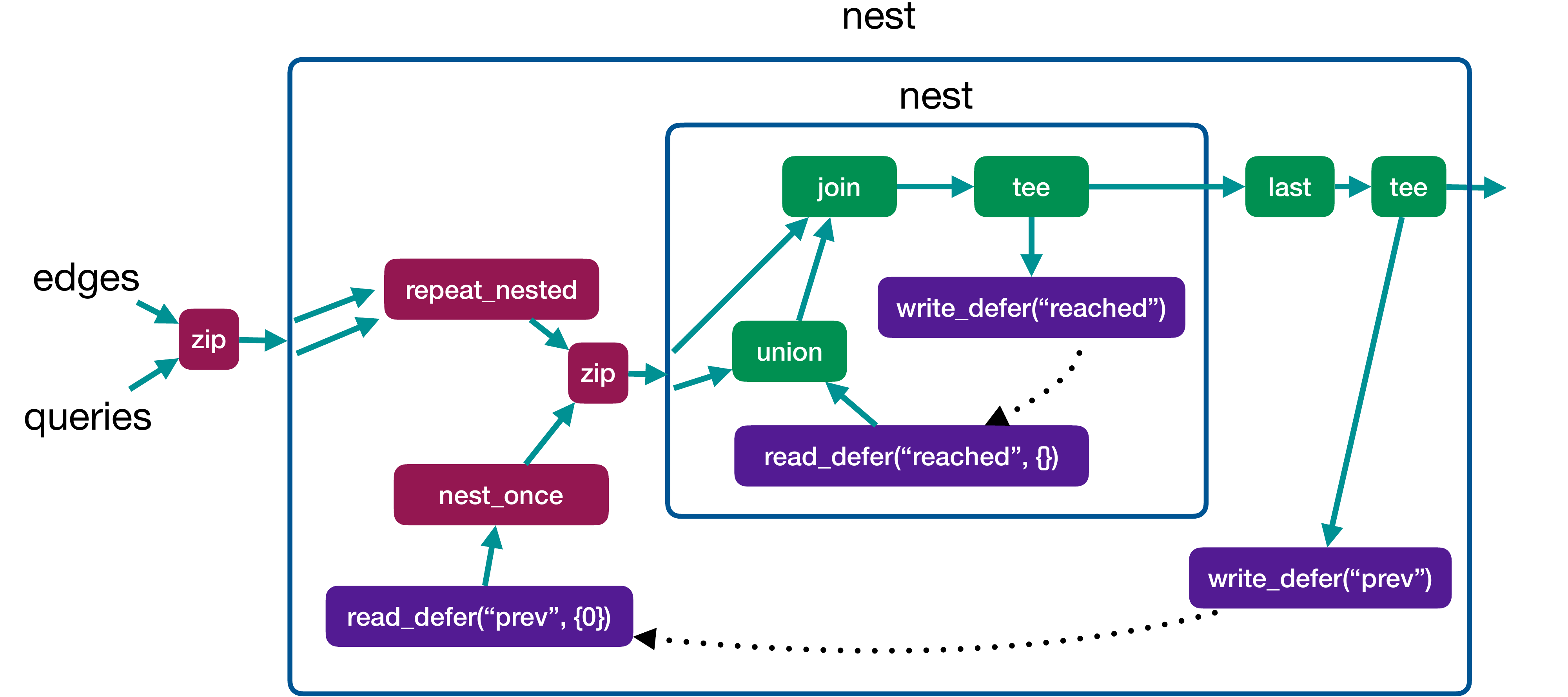}
    \caption{An example of graph reachability with a dynamic radius, using nested cycles.}
    \label{fig:graph_reachability_nested_example}
\end{figure}

In this example program, we assume that the input edges have already been shaped into an unbounded stream-of-streams where each inner stream contains the full set of edges\footnote{We could also consume the set of edges only once and ``persist'' them across iterations of the nested graph by sending a copy across a \texttt{defer} cycle. But that adds complexity to this example that distracts from nested cycles.}. The queries, which represent expansions of the radius, are also a stream-of-streams where each inner stream is a singleton containing the amount to expand the radius by. We use a new \texttt{zip} operator to feed multiple nested streams into \texttt{nest} by tupling their inner streams pairwise.

We use a new \texttt{last} operator to extract the final set emitted by reachability, which we defer to bootstrap the next query. To inject these nodes, we use a new operator \texttt{nest\_once} which generates an infinite stream-of-streams where the first inner stream contains the input and the rest are empty. Then, inside the reachability graph, we use \texttt{union} (which performs set union) to add the bootstrap nodes. Finally, we use \texttt{repeat\_nested} as before to drive iterations of graph reachability.

\subsection{Type Semantics}
Now, we are ready to lay out the formal semantics for nested graphs, beginning with the type semantics. First, we define the \texttt{defer} operators: \texttt{write\_defer} takes a key as a parameter and accumulates a bounded stream as input, and on the next iteration any matching \texttt{read\_defer} with the same key will emit the accumulated collection. Type-safety for these operators is a bit more complex, since we need to ensure that there is a single \texttt{write\_defer} for each key and that the stream types being written match the types being read.

To achieve this, we introduce a new pair of contexts $R$ and $W$ to our typing rules ($\vdash$ and $\vdash^O$) which each store a map from keys to stream types. We will use context $W$ substructurally, admitting only exchange (but not weakening or contraction) on this context. When typing a nested graph, these contexts are set to (arbitrary) identical values, which enforces that the same types are written and read. On the write-side, we also enforce that each key is written exactly once by splitting the $W$ keys at each composition until there is one key isolated to each \texttt{write\_defer}. For \texttt{read\_defer}, we have two variants because the optional second parameter stores a value to be emitted.

\begin{figure}[h]
    \begin{mathpar}\small
    {
        \inferrule[sequence]
        {R; W_1 ⊢ \es₁ : (I_1 ↪ O_1, \prec_1) \\ R; W_2 ⊢ \es₂ : (I_2 ↪ O_2, \prec_2) \\ O_1 \leq I_2}
        {R; W_1, W_2 ⊢ \es₁;\es₂ : (I_1 ↪ O_2, \prec_1)}
    }

    {   
        \inferrule[par]
        {R; W_1 ⊢ \es₁ : (I_1 ↪ O_1, \prec_1) \\ R; W_2 ⊢ \es₂ : (I_2 ↪ O_2, \prec_2)}
        {R; W_1, W_2 ⊢ \es₁ ∣ \es₂ : (I_1 \cdot I_2 ↪ O_1 \cdot O_2, \prec_1 \cdot \prec_2)}
    }

    {
        \inferrule[operator]
        {R;W \vdash^O op : (I ↪  O, \prec)
        \\
        I = ((S_0, B_0), \ldots (S_n, B_n))
        \\ \forall{i}.~\mathit{type}^C(s_i) = S_{i}}
        {R;W ⊢ \{(s_0, \ldots, s_n)\}[op] : (I ↪ O, (\prec))}
    }
    
    {
        \inferrule[read-defer-value-type]
        {\mathit{type}^C(v) = C \\ \mathit{fixed}(\llbracket{v}\rrbracket^C)}
        {R, k:C;\emptyset \vdash^O \text{read\_defer(k, v)}: (() \hookrightarrow (C, \mathbf{B}), \emptyset)}
    }

    {
        \inferrule[read-defer-no-value-type]
        {}
        {R, k: C; \emptyset \vdash^O \text{read\_defer(k)}: (() \hookrightarrow (C, \mathbf{B}), \emptyset)}
    }

    {
        \inferrule[write-defer-type]
        {}
        {R; k: C \vdash^O \text{write\_defer(k)}: ((C, \mathbf{B}) \hookrightarrow (), \emptyset)}
    }

    {
        \inferrule[nest-type]
        {D; D \vdash g: (I \hookrightarrow (O_1, \ldots), \prec_g) \\
        \forall{i}.~O_i = (C_i, \mathbf{B})}
        {R;\emptyset \vdash^O \text{nest}(g): (([I], X) \hookrightarrow ([(O_1, \ldots)], X), \prec_{\text{nest}}(\prec_g))}
    }
    
    {
        \inferrule[nest-with-copy-type]
        {D; D \vdash g: (I \hookrightarrow (O_1, \ldots O_m), \prec_g) \\
        D; D \vdash g_o: (I \hookrightarrow (O_1, \ldots O_m), \prec_g) \\
        O_i = (C_i, \mathbf{B})}
        {R;\emptyset \vdash^O \text{nest}(g, g_o): (([I], X) \hookrightarrow ([(O_1 \ldots O_m)], X), \prec_{\text{nest}}(\prec_g))}
    }
    \end{mathpar}
    
    \caption{Type semantics with defer contexts, and for \texttt{read\_defer}, \texttt{write\_defer}, and \texttt{nest}.}
    \label{fig:nested_type_semantics}
\end{figure}

The \texttt{nest} operator takes a graph $g$ of type $I \hookrightarrow O$ with partial order $\prec_g$. Each stream in $O$ must be \textbf{bounded} so that the inner graph finishes in finite time. The operator itself takes a stream of streams and emits a stream of streams, where the inner types are $I$ and $O$ respectively. The boundedness of the outer output (denoted $X$) is the same as the outer input. We also include a variant of \texttt{nest} with an additional parameter that stores the initial graph for the next iteration. We re-define our core composition type semantics with these contexts as well as for \texttt{write\_defer}, \texttt{read\_defer}, and \texttt{nest} in \cref{fig:nested_type_semantics}. Note that this requires a modification to the full type system; we do this in the usual way.  In particular, note that as existing operators never have graphs as subterms, they will be lifted into our context-enhanced system with arbitrary $R$ and empty $W$ contexts.

\subsection{Operational Semantics}

The \texttt{nest} operator processes tuples of inner streams one-by-one, maintaining the current inner streams at the rightmost element of the input. It shifts to the next tuple of inner streams once the graph reaches a stuck state and all the outputs (including those to \texttt{write\_defer}) are fixed. The \texttt{nest} operator first stores a copy of the initial graph as a second parameter (this variant is lower in the partial order for \texttt{nest}). To process an inner stream, we use $\mathit{setinput}$ to set the inner graph inputs, step the inner graph, and then use $\mathit{inputs}$ to propagate input consumption to the nested stream. Once the input only contains a terminator, the operator emits a terminator as well.

Note that \texttt{write\_defer} has no small-step rules; its behavior is handled by the semantics for \texttt{nest}. The \texttt{read\_defer} operator takes a single small-step, which emits its collection parameter. This collection parameter is either a default value (for the first tuple of inner streams) or a value from \texttt{write\_defer}. When shifting to the next inner stream input, we use the $\mathit{collect\_defer}$ helper to accumulate the inputs to each \texttt{write\_defer} into a map, and then use the $\mathit{set\_defer}$ helper to create a copy of the initial graph with the corresponding \texttt{read\_defer} operators updated to use those collections. We visualize this behavior in \cref{fig:nest_operators_viz} where a stream-of-streams on the left, with later elements lower, is transformed into another stream-of-streams. We then lay out the formal operational semantics in \cref{fig:nest_operational}.

\begin{figure}[h]
    \centering
    \includegraphics[page=2,trim={10cm 4.6cm 10cm 1.3cm},clip,width=0.6\textwidth]{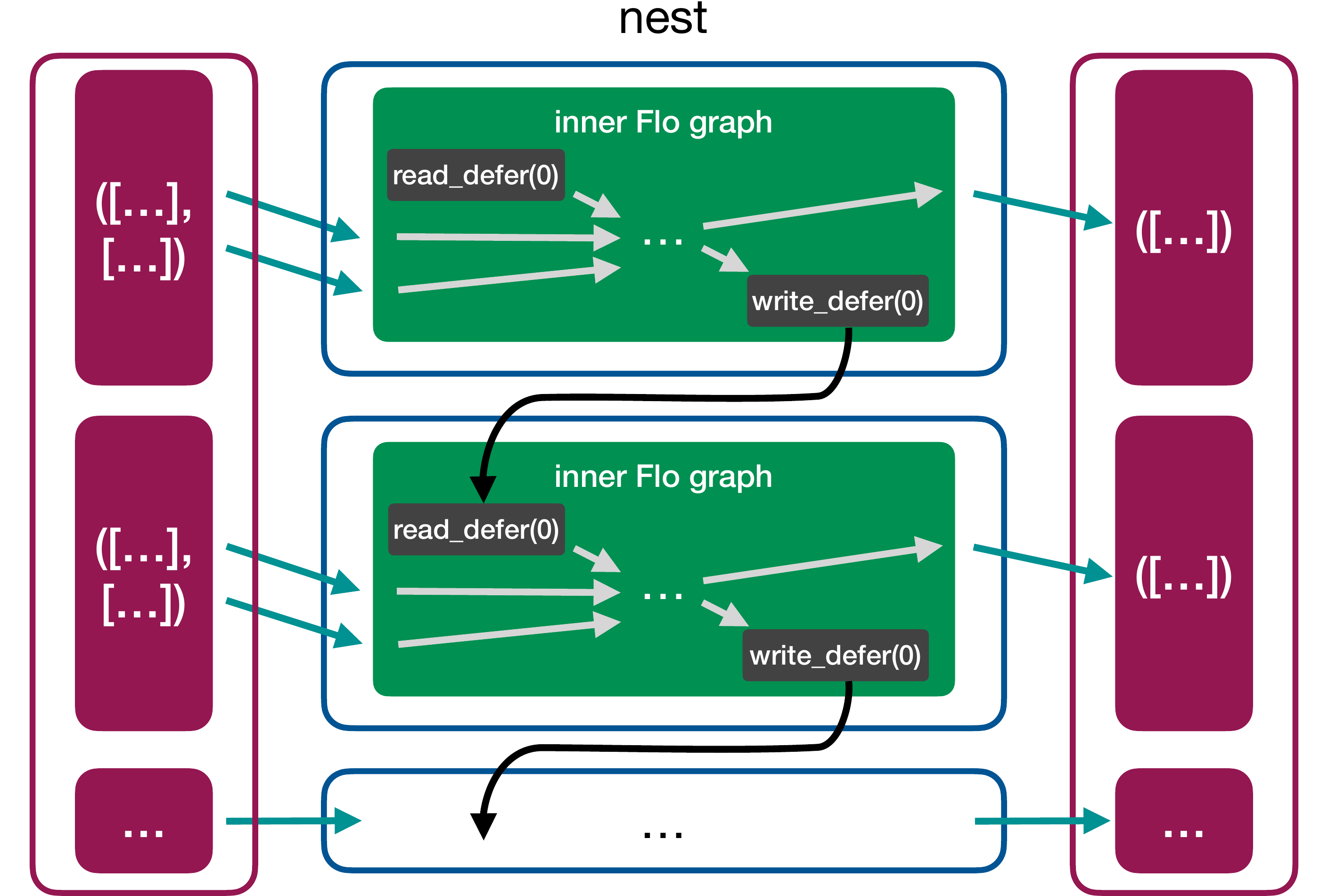}
    \caption{Visualization of the \texttt{nest}, \texttt{read\_defer}, and \texttt{write\_defer} operators, where the nested streams on the left and right have later elements lower.}
    \label{fig:nest_operators_viz}
\end{figure}

\begin{figure}[H]
    \newcommand{\deferin}{\ensuremath\textit{collect\_defer}}
    \newcommand{\setdeferout}{\ensuremath\textit{set\_defer}}
    
    \small
    \begin{align*}
        \deferin(\es₁;\es₂) ≜&~\deferin(\es₁) \cup \deferin(\es₂)\\
        \deferin(\es₁∣\es₂) ≜&~\deferin(\es₁) \cup \deferin(\es₂)\\
        \deferin(\{I\}[\mathit{write\_defer(k)}]) ≜&~\{ k: I \}\\
        \deferin(\{I\}[op]) ≜&~\terminator&\textit{when}~op \neq \mathit{write\_defer}
    \end{align*}
    \vspace{-0.3cm}
    \begin{align*}
        \setdeferout(\es₁;\es₂,M) ≜&~\setdeferout(\es₁,M);\setdeferout(\es₂,M)&\\
        \setdeferout(\es₁∣\es₂,M) ≜&~\setdeferout(\es₁,M)∣\setdeferout(\es₂,M)&\\
        \setdeferout(\{\}[\mathit{read\_defer}(k, v)],M) ≜&~\{\}[\mathit{read\_defer}(k, M[k])]\\
        \setdeferout(\{I\}[op],M) ≜&~\{I\}[op]&\textit{when}~op \neq \mathit{read\_defer}
    \end{align*}
    
    \begin{mathpar}\small
    {
        \inferrule[nest-first]
        {I \neq \terminator}
        {([\ldots, I], \text{nest}(g)) \rightarrow^\delta ([\ldots, I], \text{nest}(g, g), ((\bot, \ldots, \bot), true))}
    }
    
    {
        \inferrule[nest-first-fixed]
        {}
        {([\terminator], \text{nest}(g)) \rightarrow^\delta ([\terminator], \text{nest}(g, g), \terminator)}
    }
    
    {
        \inferrule[nest-run-graph]
        {(\mathit{setinput}(g, \lfloor{I}\rfloor^C)) \rightarrow^\Delta (g', (O_1', \ldots, O_m'))}
        {([\ldots, I], \text{nest}(g, g_o)) \rightarrow^\delta ([\ldots, \llbracket\mathit{inputs(g')}\rrbracket^C], \text{nest}(g', g_o), ((O_1', \ldots, O_m'), false))}
    }
    
    {
        \inferrule[nest-run-step]
        {(\mathit{setinput}(g, \lfloor{I}\rfloor^C), (O_1, \ldots, O_m)) \text{ is stuck} \\
        \forall{m}. \mathit{fixed}(O_m) \\
        \forall_{d \in \deferin(g)}~\mathit{fixed}(d) \\
        I_{next} \neq \terminator}
        {([\ldots, I_{next}, I], \text{nest}(g, g_o)) \rightarrow^\delta ([\ldots, I_{next}], \text{nest}(\setdeferout(g_o, \deferin(g)), g_o), ((\bot, \ldots, \bot), true))}
    }
    
    {
        \inferrule[nest-run-fixed]
        {(\mathit{setinput}(g, \lfloor{I}\rfloor^C), (O_1, \ldots, O_m)) \text{ is stuck} \\
        \forall{m}. \mathit{fixed}(O_m)}
        {([\terminator, I], \text{nest}(g, g_o)) \rightarrow^\delta ([\terminator], \text{nest}(g_o, g_o), \terminator)}
    }
    
    {
        \inferrule[read-defer-emit]
        {}
        {((), \texttt{read\_defer}(k, v)) \rightarrow^\delta ((), \texttt{read\_defer}(k), v)}
    }
    \end{mathpar}
    
    \caption{Small-step semantics for the \texttt{nest} and \texttt{read\_defer} operators.}
    \label{fig:nest_operational}
\end{figure}

\subsection{Operator Properties}
Because \texttt{nest} is a standard operator, it must satisfy all Flo's core operator properties. First, we define the partial order $\prec_{\texttt{nest}}(\prec_g)$, which is parameterized over the partial order for the inner graph. Our small step semantics either consume the rightmost input inner stream or reduce it according to the nested graph's partial order. So we have $$[...] \prec_{\texttt{nest}}(\prec_g) [..., I]$$
$$[..., I'] \prec_{\texttt{nest}}(\prec_g) [..., I] \text{ if } I' \prec_g I$$
$$\terminator \prec_{\texttt{nest}}(\prec_g) [...]$$

For \texttt{read\_defer}, any operator expression \emph{without} the value parameter is smaller than any with it, so the step for \texttt{read\_defer} reduces the operator expression. Since \texttt{write\_defer} takes no steps, it satisfies our operator proof obligations trivially. We can now prove the properties of \texttt{nest}.

\textbf{Operator Well-Formedness:}
\begin{proof}
When we step across an input ({\sc nest-run-step} and {\sc nest-run-fixed}), the input is updated to a prefix, which satisfies our first case of the partial order. The only other rule that modifies inputs is {\sc nest-run-graph}, which will only touch the inputs if it recursively steps a left-most operator that consumes those inputs. Because of \cref{lem:operator_stuck}, we know that running any of these operators will reduce the input along the partial order for the inner graph.
\end{proof}

\textbf{Operator Preservation:}
\begin{proof}
There are only two ways we modify the inputs and outputs; either we push or pop an entire tuple of inner streams or update the rightmost input or leftmost output. In the first case, we only push $\bot$, and popping does not affect the type of the collection. When we update an input/output instead, \cref{lem:graph_preservation} guarantees that this is safe. In all our rules, the operator is only changed by setting the inputs of the graph, which is safe because the input types are unchanged.
\end{proof}

\textbf{Operator Determinism:}
\begin{proof}
First, {\sc nest-first} or {\sc nest-first-fixed} will execute, then {\sc nest-run-graph} will run until stuck state, then {\sc nest-run-step} will run, until the input stream is fixed and {\sc nest-run-fixed} is run. In {\sc nest-run-graph}, the only rule where we recursively apply a step, we know that the stuck state exists (\cref{lem:graph_stuck}) and is deterministic (\cref{lem:graph_determinism_eager}). Therefore, \texttt{nest} is deterministic.
\end{proof}

\textbf{Eager Execution:}
\begin{proof}
For {\sc nest-first-fixed} and {\sc nest-run-fixed}, because the input collection is already fixed deltas have no effect. For {\sc nest-first}, regardless of whether the delta is introduced before or after, the final state will be the same because we copy the input as-is and a concatenation will never affect $I \neq \terminator$ because an element can never be replaced by the terminator. Because {\sc nest-run-graph} will run until the inner graph reaches a stuck state, we can apply \cref{lem:graph_determinism_eager} to know that introducing a delta before or after the step will result in the same final state, because introducing a delta to the nested stream will only affect the last element $I$. For {\sc nest-run-step}, the delta will never affect $I$, and any delta to $I_{next}$ will be applied the same before or after the step.
\end{proof}

\textbf{Streaming Progress:}
\begin{proof}
If the input is bounded, the output will become fixed because each iteration will finish in finite time by \cref{lem:graph_streaming_progress}. If it is unbounded, the input sequence being fixed only affects {\sc nest-first-fixed} and {\sc nest-run-fixed} rules, which simply concatenate a terminator to the output sequence without modifying it in any other way.
\end{proof}

\section{Case Studies}
\label{sec:case_studies}

Flo aims to provide strong guarantees that are meaningful across a range of applications while remaining sufficiently abstract to capture a variety of semantics. In this section, we demonstrate the expressiveness of Flo by using it to implement the key ideas found in existing streaming languages. Note that our goal is \textbf{not} to show how to implement these entire languages in Flo, rather that key ideas from them can be expressed and satisfy our properties. We focus on three existing languages:
\begin{enumerate}
\item Flink~\cite{flink}, a popular streaming framework that features windowed aggregation functions.
\item LVars~\cite{lvars}, a language for parallel programming that uses lattices to ensure determinism.
\item DBSP~\cite{dbsp}, a system for incremental view maintenance that uses z-sets to model relations.
\end{enumerate}

\subsection{Flink}
Flink~\cite{flink} is a classic example of a streaming dataflow language. Like Flo, Flink uses compositions of operators to describe computations over streams. A key technique from Flink is the use of \emph{windows} to enable aggregations over fixed-time intervals of an infinite stream. We will show that this idea from Flink can be modeled in Flo as a timestamped collection type, where windowing operators generate streams-of-streams which can then be aggregated in a nested graph.

Flink uses ordered sequences as its primary semantics for streams. We can model this in Flo by introducing an ordered sequence collection, which simply stores a list of values where the newest items are on the left and the oldest elements on the right. We define this collection in \cref{fig:ordered_sequence_type_semantics}.

\begin{figure}[h]
$$\texttt{S<V>} \triangleq \{ [v_1, \ldots, v_n]~|~\forall{i}.~v_i \in V \} \cup \{ [\terminator, v_1 \ldots v_n]~|~\forall{i}.~v_i \in V \}$$
$$[v_1, \ldots, v_n] \concat [d_1, \ldots, d_m] = [d_1, \ldots, d_m, v_1, \ldots, v_n]$$
$$[\terminator, \ldots] \concat x = [\terminator, \ldots]$$
$$[v_1, \ldots, v_n] \concat \terminator = [\terminator, v_1, \ldots, v_n]$$

\caption{Collection type and concatenation operator for ordered sequences in Flo.}
\label{fig:ordered_sequence_type_semantics}
\end{figure}

With a collection type for ordered sequences, we can define classic operators found in Flink such as \texttt{map}. We can also define semantics for \texttt{fold} that matches the Flink semantics of emitting a stream containing a single value, which is the result of the aggregation. We can define the type and operational semantics for these operators in \cref{fig:flink_operational_semantics} (we omit partial orders for brevity).

\begin{figure}[h]
\begin{mathpar}\small
{
    \inferrule[map-type]
    {\vdash f: T \rightarrow U}
    {\vdash^O \text{map}(f): ((\texttt{S<T>}, X) \hookrightarrow (\texttt{S<U>}, X), \prec_{\text{map}})}
}

{
    \inferrule[map]
    {f(h) \Downarrow u}
    {([\ldots, h], \text{map}(f)) \rightarrow^\delta ([\ldots], \text{map}(f), [u])}
}

{
    \inferrule[map-terminator]
    {}
    {([\terminator], \text{map}(f)) \rightarrow^\delta (\terminator, \text{map}(f), \terminator)}
}

{
    \inferrule[fold-type]
    {\vdash acc: U\\\vdash f: (U, T) \rightarrow U}
    {\vdash^O \text{fold}(acc, f): ((\texttt{S<T>}, B) \hookrightarrow (\texttt{S<U>}, B), \prec_{\text{fold}})}
}

{
    \inferrule[fold]
    {f(acc, h) \Downarrow acc'}
    {([\ldots, h], \text{fold}(acc, f)) \rightarrow^\delta ([\ldots], \text{fold}(acc', f), [])}
}

{
    \inferrule[fold-terminator]
    {}
    {([\terminator], \text{fold}(acc, f)) \rightarrow^\delta (\terminator, \text{fold}(acc, f), [\terminator, acc])}
}
\end{mathpar}

\caption{Operational semantics for Flink operators in Flo.}
\label{fig:flink_operational_semantics}
\end{figure}

Map processes elements one by one and passes through the terminator, so it satisfies eager execution and streaming progress easily. But note that for the fold operator to satisfy streaming progress, its input must be bounded, otherwise the step that emits the aggregated value when the input becomes fixed would be illegal.

Now given an unbounded stream, how do we use fold? Flink's answer is to use windows, where the aggregation is run over blocks of data defined by timestamp intervals. This idea maps perfectly to the Flo model, where we can convert an unbounded stream of timestamps-value pairs into a stream-of-streams (as in \cref{sec:nested}) and then use a nested graph to aggregate over each window.

To implement this windowing operator, we will use the internal state of the operator to store the values corresponding to the next window. When a timestamp farther than the end of the current interval is received, we emit the accumulated window. Because the operator uses timestamp boundaries to determine when to emit inner streams, the inner streams are bounded even though the outer stream-of-streams is unbounded. We omit detailed proofs for brevity, but this operator also satisfies eager execution and streaming progress. We can sketch the type and operational semantics for this operator in \cref{fig:window_operational_semantics} (again omitting the partial order for brevity).

\begin{figure}[h]
\begin{mathpar}\small
{
    \inferrule[window-type]
    {interval \text{ is an amount of time}\\T \text{ is a timestamp}\\\forall{i}~t_i\text{ is a timestamp}\\\forall{i}~\vdash v_i: D}
    {\vdash^O \text{window}(interval, [(v_1, t_1), \ldots (v_n, t_n)]): ((\texttt{S<(D, T)>}, X) \hookrightarrow ([(\texttt{S<D>}, B)], X), \prec_{\text{window}})}
}

{
    \inferrule[window-first]
    {}
    {([\ldots (v_n, t_n)], \text{window}(interval, [])) \rightarrow^\delta ([\ldots], \text{window}(interval, [(v_n, t_n)]), [])}
}

{
    \inferrule[window]
    {t_n - wt_m \leq \text{interval}}
    {([\ldots (v_n, t_n)], \text{window}(interval, [(w_1, wt_1), \ldots, (w_m, wt_m)])) \rightarrow^\delta\\ ([\ldots], \text{window}(interval, [(v_n, t_n), (w_1, wt_1), \ldots, (w_m, wt_m)]), [])}
}

{
    \inferrule[window-emit]
    {t_n - wt_m > \text{interval}}
    {([\ldots (v_n, t_n)], \text{window}(interval, [(w_1, wt_1), \ldots, (w_m, wt_m)])) \rightarrow^\delta\\ ([\ldots], \text{window}(interval, [(v_n, t_n)]), ([w_1, \ldots, w_m], true))}
}
\end{mathpar}

\caption{Type and operational semantics for the \texttt{window} operator.}

\label{fig:window_operational_semantics}
\end{figure}

To complete our example of how patterns from Flink can be modeled in Flo, we can perform aggregations over these windows by using a nested graph. We can pass the result of the window operator into the \texttt{nest} operator defined in \cref{sec:nested}, and use the \texttt{fold} operator inside the nested graph. Because the nested stream is bounded, this will typecheck and the aggregation will be appropriately computed for each window.

\subsection{LVars}
LVars~\cite{lvars} is a language for deterministic parallel programming that uses lattice-based data structures to ensure determinism. A key insight of LVars is to leverage monotonicity to ensure determinism, by requiring that pieces of state are always updated monotonically, and restricting reads of the state to threshold queries that check if the state is larger than a given value. We will show that the essence of LVars can be modeled in Flo as a special collection type, where threshold queries can be used to safely read from lattice values that are derived from unbounded aggregations.

First, let us define the collection for an LVar. Consider a lattice defined by a set of values $L$, a bottom value $\bot$, and the lattice join operator $\sqcup$. We will define the LVar collection type a tuple of the lattice value and a boolean flag, where the boolean flag indicates whether the value is $\mathit{fixed}$ or not. We will use the lattice join for concatenation, and the $\terminator$ terminator to terminate a collection.

\begin{figure}[h]
$$\texttt{LVar<L>} \triangleq \{ (v, \mathit{true})~|~v \in L \} \cup \{ (v, \mathit{false})~|~v \in L \}$$
$$(v_1, \mathit{false}) \concat v_2 = (v_1 \sqcup v_2, \mathit{false})$$
$$(v_1, \mathit{true}) \concat v_2 = (v_1, \mathit{true})$$
$$(v_1, \mathit{false}) \concat \terminator = (v_1, \mathit{true})$$

\caption{Type semantics and concatenation operator for LVars in Flo.}
\label{fig:lvar_type_semantics}
\end{figure}

There are many operators that can produce an LVar from various input collection types. Let us use ordered sequences as an example. We can define a \texttt{fold\_lattice} operator which transforms each value into a lattice and then applies the lattice join across the sequence. We define the type and operational semantics for this operator in \cref{fig:fold_lattice_semantics}.

\begin{figure}[h]
\begin{mathpar}\small
{
    \inferrule[fold-lattice-type]
    {\vdash f: T \rightarrow U \\
     U\text{ is a lattice}
    }
    {\vdash \text{fold\_lattice}(f): (\texttt{S<T>}, X) \hookrightarrow (\texttt{LVar<U>}, X)}
}

{
    \inferrule[fold-lattice]
    {v \neq \terminator \\ f(v) \Downarrow l}
    {([..., v], \text{fold\_lattice}(f)) \rightarrow^\delta ([...], \text{fold\_lattice}(f), l)}
}

{
    \inferrule[fold-lattice-terminated]
    {}
    {([\terminator], \text{fold\_lattice}(f)) \rightarrow^\delta (\terminator, \text{fold\_lattice}(f), \terminator)}
}
\end{mathpar}

\caption{Type and operational semantics for the \texttt{fold\_lattice} operator.}
\label{fig:fold_lattice_semantics}
\end{figure}

We omit detailed proofs of the core operator properties for brevity here, but note that the boundedness of the output is equal to the boundedness of the input. This is because we can guarantee a terminator on the output when the input will become fixed. In addition, we satisfy eager execution because we always consume elements from the rightmost side, and concatenation to the input can only introduce new elements on the left.

Consider a naive attempt to implement an operator that converts an \texttt{LVar<T>} back into an ordered sequence \texttt{[T]} by generating a stream containing a single value with that LVar: $$((v, \_), \text{to\_sequence}) \rightarrow^\delta (\terminator, \text{to\_sequence}, [v])$$

This operator will be \textbf{illegal} because it does not satisfy eager execution. Recall that we are interested in convergence regardless of whether a delta is introduced before or after the step. If we introduce a delta that changes the lattice value, the output sequence would be different depending on this scheduling decision, making the operator non-deterministic. Let's make another attempt to implement this operator, where we wait for the LVar to be fixed first: $$((v, true), \text{to\_sequence}) \rightarrow^\delta (\terminator, \text{to\_sequence}, [v])$$

This operator satisfies eager execution, but \textbf{now fails} to satisfy streaming progress when instantiated with an unbounded streaming input! If we run the operator on an unfixed input, the output will be an empty sequence. But if we terminate this input, the output will grow to include the lattice value, which is illegal because streaming progress mandates that the only change between these executions should be that the output also becomes fixed, without any changes to its contents. A fix is to restrict the typing rules for the operator to only accept bounded inputs, so that the input is guaranteed to be eventually fixed.

What can we do with \emph{unbounded} LVars? The fundamental properties of Flo and the original LVars paper come to the same conclusion: we must use a threshold query instead. We can define an operator that takes an LVar and a threshold value, and emits the threshold if the input exceeds it. We list the type and operational semantics for this operator in \cref{fig:threshold_semantics} (omitting partial orders).

\begin{figure}[h]
\begin{mathpar}\small
{
    \inferrule[lvar-threshold-type]
    {\forall{i}.~t_i \in U \land \forall{i,j}~\nexists{k}.~i \neq j \land t_i \sqcup t_j = k}
    {\vdash_O \text{thresh}(t_1,~\ldots): ((\texttt{LVar<U>}, X) \hookrightarrow (\texttt{S<U>}, X), \prec_{\text{thresh}})}
}

{
    \inferrule[lvar-threshold]
    {v \sqcup t_i = v}
    {((v, \_), \text{thresh}(t_1,~\ldots)) \rightarrow^\delta (\terminator, \text{thresh}(t_1,~\ldots), t_i)}
}

{
    \inferrule[lvar-threshold-terminated]
    {}
    {((v, true), \text{thresh}(\ldots)) \rightarrow^\delta (\terminator, \text{thresh}(\ldots), \terminator)}
}
\end{mathpar}

\caption{Type and operational semantics for the \texttt{threshold} operator.}
\label{fig:threshold_semantics}
\end{figure}

This operator satisfies \textbf{both} eager execution and streaming progress, making it safe to use in a Flo program. The more general properties required for Flo, which do not involve partial orders over collection values or any algebraic properties, still map very precisely to the approach taken in LVars to enable deterministic data processing.

\subsection{DBSP}
\label{sec:dbsp}
Another point in the streaming language design space comes from the database community. DBSP~\cite{dbsp} introduces a formal model for relational operators that can be incrementally executed on live updating databases. A key insight of DBSP is that relations with incremental updates can be modeled as z-sets, where each element in the set has an integer cardinality, such that negative values correspond to retractions of data. We will show that the essence of DBSP can be modeled in Flo by using a special collection type for z-sets, where incremental operations over these correspond to satisfying eager execution.

First, let us define the collection for a z-set in \cref{fig:zset_type_semantics}. We will define the z-set collection type as a map of keys to integer cardinalities as well as a boolean flag that indicates that the collection is fixed. The concatenation operator simply combines the two maps by adding the cardinalities of matching keys, and the $\terminator$ terminator makes the collection fixed.

\begin{figure}[h]
Cardinality Maps: $M = \{ k_1: v_1, \ldots \} \text{ where } v_i \in \mathbb{Z}$, $M[k] = 0 \text{ if } k \notin M$
$$(M_1 + M_2)[k] = M_1[k] + M_2[k]$$

ZSet = $\{(m, \mathit{true})~|~m \in M\} \cup \{(m, \mathit{false})~|~m \in M\}$
\begin{align*}
    (M_1, false) \concat M_2 &= \{ (M_1 + M_2, false) \} \\
    (M, \_) \concat \terminator &= \{ (M, true) \}
\end{align*}

\caption{Collection type and concatenation operator for z-sets in Flo.}
\label{fig:zset_type_semantics}
\end{figure}

In DBSP, inputs to the program are z-sets, and we will take the same approach when mapping this to Flo. Next, we define operators over z-sets. Let us define \texttt{map}, a general version of projection, in \cref{fig:zset_basic_operational}. We omit typing rules for brevity, but the output boundedness is the same as the input.

\begin{figure}[h]
\begin{mathpar}\small
{
    \inferrule[map-zset]
    {f(k_1, v_1) \Downarrow v'}
    {((\{ k_1: v_1, \ldots\}, \_), \text{map}(f)) \rightarrow^\delta\\\\((\{\ldots\}, \_), \text{map}(f), (\{ k_1: v' \}, \_))}
}

{
    \inferrule[map-zset-terminated]
    {}
    {((\{\}, true), \text{map}(f)) \rightarrow^\delta (\terminator, \text{map}(f), \terminator)}
}
\end{mathpar}

\caption{Small-step semantics for the map operator.}
\label{fig:zset_basic_operational}
\end{figure}

This operator trivially satisfies streaming progress, because no outputs are gated on termination. In DBSP, the primary goal is incremental execution: we can introduce additional input and the output will be updated to the result on the full input. This is \emph{exactly} the definition of \textbf{eager execution}. Our operators satisfy this property because they are distributive over the z-set. Consider processing a key $k_1$ with cardinality $v_1$ only to have it re-introduced by a delta with cardinality $v_2$. If the delta is applied before the operator, the operator will directly emit a value with cardinality $v_1 + v_2$. If the delta is applied after, cardinality $v_1$ will be emitted, and later the operator will emit $v_2$ which will be added together by concatenation.

\begin{figure}[h]
$$(M_1 \bowtie M_2)[k] = M_1[k] \cdot M_2[k]$$
\begin{mathpar}\small
{
    \inferrule[join-zset]
    {}
    {((M_1', s_1), (M_2', s_2), \bowtie(M_1, M_2)) \rightarrow^\delta\\((\{\}, s_1), (\{\}, s_2), \bowtie(M_1 + M_1', M_2 + M_2'), (M_1 \bowtie M_2' + M_1' \bowtie M_2 + M_1' \bowtie M_2'))}
}

{
    \inferrule[join-zset-terminated]
    {}
    {((\{\}, true), (\{\}, true), \bowtie(\_, \_)) \rightarrow^\delta (\terminator, \terminator, \bowtie, \terminator)}
}
\end{mathpar}

\caption{Operational semantics for the join operator.}
\label{fig:zset_join_operational}
\end{figure}

A more interesting operator is the natural join ($\bowtie$), which takes two z-sets and produces a new z-set by joining on a key. First, we define a $\bowtie$ operator on z-sets which joins them by taking the product of cardinalities of matching keys. To perform an incremental join, we store the z-sets which have already been processed in the state of the operator. We can then apply the z-set property $(a + a') \bowtie (b + b') = a \bowtie b + a' \bowtie b + a \bowtie b' + a' \bowtie b'$. We use this in a sketch for the operational semantics in \cref{fig:zset_join_operational} (again, omitting type semantics but using only unbounded streams).

Again, what is interesting here is that proving eager execution \emph{aligns exactly} with the incremental computation goal in DBSP. In DBSP, proofs of correctness hinge on the join operator being bilinear, because $a \cdot (b + c) = a \cdot b + a \cdot c$. This is exactly the property we need to prove eager execution, because the operator must be distributive over concatenations to the z-set. This is a powerful demonstration of the flexibility of Flo, as it can precisely capture the semantics of incremental computation with retractions, a key limitation of approaches like Stream Types~\cite{streamtypes}.

\subsection{Putting it Together}
What is particularly exciting is that all these case studies fit into the common model of Flo. In fact, we could unify all three into a single language, since the operators are all composable and can be used together. For example, we shared the ordered sequence collection between Flink and LVars, so the operators we defined in both could easily be mixed together to compute a threshold over windowed aggregates. This shows the power of the abstract approach taken by Flo; we can capture a wide range of semantics under one roof, while still providing strong guarantees about the behavior of the system as a whole.

\section{Related Work}
Flo builds on the vast bodies of work on streaming language design from both the programming languages and databases communities. We leverage the insights across both traditional stream processing and incremental computation to devise a new model for progressive streams.

\subsection{Stream Types and Deterministic Dataflow}
The most closely related work to Flo recently is Stream Types~\cite{streamtypes}, which provides a rich type system that can precisely capture the structure of elements in a stream. Stream Types focus on capturing ordering invariants, such as the presence of certain elements within bracketing pairs. These fine-grained types make it possible to prove strong semantic guarantees about the \emph{implementation of operators}, such as determinism when operating on prefixes of data.

These properties map well to the eager execution and streaming progress properties of Flo, which takes a more abstract compositional approach to stream semantics. In this way, Stream Types and Flo can be complementary, since Stream Types can be used to prove that operators in a Flo language satisfy the properties required by Flo.  Flo's notion of streams, however, is more general than that of Stream Types; indeed, one of the key limitations of Stream Types is that they cannot model incremental computation with retractions, a key feature of DBSP that Flo can capture.

Other work defines streams as monoids~\cite{mamouras2020semantic, data_trace_types} and uses monotone operators to ensure determinism. We generalize this approach by relaxing their monotonicity requirements into eager execution, and by relying on a notion of concatenation that generalizes their monoidal structure. This enables Flo to be used to model retractions that the monoidal approach cannot capture.  

\subsection{Stream Query Languages}
``Continuous'' query languages over streams have been a topic of recurring interest in database research since the 1990s. A recent tutorial article overviews the history of that work~\cite{carbone2020beyond}, and highlights the foundational influence of CQL~\cite{arasu2006cql} on language semantics. CQL extends SQL with operators that map a family of timestamped stream collections (unbounded, in our terminology) to relations (bounded) and vice-versa; SQL is used as an inner language to map relations to relations. CQL assumes a totally ordered, timestepped model of execution in which all data for each timestep is known to be available when that timestep is processed. Like many stream query languages of its time, CQL does not address delay directly: ``Our semantics does not dictate `liveness' of continuous query output---that issue is relegated to latency management in the query processor~\cite{babcock2003chain,carney2003operator}''.

The same tutorial also points out various constructs that stream query languages introduced for \emph{tracking} progress, including punctuations~\cite{tucker2003exploiting}, watermarks~\cite{akidau2015dataflow}, heartbeats~\cite{srivastava2004flexible}, slack~\cite{abadi2003aurora}, and frontiers~\cite{naiad}. While some of these are operational (e.g., timeout-based), many fit our framework in two places: families of collection types that admit reasoning about fixedness (e.g., mixing data and control messages), and language constructs for extracting bounded ``inner'' collections.

An additional recurring discussion in these systems relates to the practical issue of ``late-arriving information'' or ``out of order processing,'' in which input values arrive that require a system to ``compensate'' for or ``retract'' previously-emitted output values. As illustrated in Section~\ref{sec:dbsp}, recent approaches~\cite{dbsp,naiad} show how these concerns can be made orthogonal to our discussion here by lifting compensations and their handling into richer collection types and operator algebras.

\subsection{Streaming Dataflow Systems}
There has been much work on building \emph{performant} streaming dataflow systems, particularly for use in analytical workloads. Systems like Samza~\cite{samza}, Storm~\cite{storm}, Flink~\cite{flink}, Heron~\cite{heron}, Beam~\cite{beam}, and Spark Streaming~\cite{discretized_streams} all provide complete systems for stream dataflow. These systems are highly performant, and as a result, they focus on the operational aspects of streaming systems, such as fault tolerance, scalability, and low-latency processing. As such, many of the contributions of these systems center on managing persistence of data on distributed nodes and preservation of deterministic outputs in the face of failures, an operational concern that we abstract away in Flo.

More recent work has focused on batching as a way to improve performance~\cite{arc,weld}, which can be modeled in Flo using nested streams. All these approaches, however, generally focus on ordered sequences as a global stream type, rather than allowing programs to mix and match collection types as in Flo. Although Flo is a theoretical foundation, we believe there is much work to be done in building a practical streaming system that can leverage the guarantees provided by Flo.

\subsection{Reactive, Incremental, and Stream-based Programming}

Much work exists on functional reactive programming (FRP), a paradigm
in which programs are continuously re-run (often incrementally) on
ever-changing inputs \cite{LTL-types-FRP, paykin2016essence, fair-rp,
fr-types, frp-without-leaks}.  These programs can be formalized as
streams, and are often compiled to a streaming dataflow representation
similar to those we explore in this paper.  Of particular interest are
papers which reason about avoiding space-time leaks \cite{fr-types,
frp-without-leaks}, requiring a property similar to our streaming
progress condition.

Other work in this space has focused on the
correspondence between LTL and FRP \cite{LTL-types-FRP,
paykin2016essence, fair-rp}, or have focused on the incrementalization
of functional programs \cite{inc, adapton}.  While our work also
reasons about properties like equivalence under re-ordering, eventual
termination, and avoiding space-time leaks, we choose a new, more
general formalism both better-suited to our domain and less
opinionated about the definitions of ``streams'' and ``operators.''

Many stream-based languages have precise ideas of how to
define both streams and computations
\cite{datafun, crdt-stream-processing, streamit, berry1985esterel, lustre}.
While much of this work is interested in properties similar to eager
execution and streaming progress, all of it is formalized with a
syntax and semantics for a particular language.  In contrast,
Flo offers an abstract, general framework for streaming languages,
with only enough constraints to prove our core properties.  We believe
that Flo provides a basis to build such languages.

An incremental, streaming language of particular interest is Naiad~\cite{naiad}, which uses a dataflow model that supports incremental execution of dataflow with cycles. Our model of nested streams is inspired by Naiad, which similarly uses special operators to describe how streams are fed into out of nested loops. In Flo, our collection type for ordered sequences of streams requires inner streams to be processed in-order, while Naiad allows for ``time-travelling'' with vector timestamps to allow modifications to already-processed streams. One could imagine implementing this in Flo using a specialized collection type and nesting operator for timestamped messages.

Other work in the streaming space focuses on a similar goal of unifying several streaming semantics under one language~\cite{soule2010universal}. But this work makes limited guarantees about the behavior of the program, with respect to both correctness and liveness of outputs. Flo provides a similar general model, but supports compositional proofs of determinism and completeness of outputs.

\section{Conclusion}
In this paper, we introduced Flo, a parameterized streaming dataflow language that provides strong guarantees about the behavior of streaming computations. Flo identifies two key properties which are general yet necessary for streaming programs: \textbf{streaming progress} and \textbf{eager execution}. We formally model these properties and show that they are preserved across composition. Furthermore, we showed that Flo supports nested streams and graphs while maintaining the semantic guarantees of the language. To demonstrate the capabilities of Flo, we showed that Flo can capture a wide range of streaming semantics, from windowed aggregation in Flink, to monotone thresholds in LVars, and even incremental computation in DBSP. We believe that Flo provides a powerful foundation for building streaming systems that can be used to more strongly reason about their guarantees.

\begin{acks}
We thank our anonymous reviewers for their insightful feedback on this paper. This work is supported in part by National Science Foundation CISE Expeditions Award CCF-1730628, IIS-1955488, IIS-2027575, DOE award DE-SC0016260, ARO award W911NF2110339, and ONR award N00014-21-1-2724, and by gifts from Amazon Web Services, Ant Group, Ericsson, Futurewei, Google, Intel, Meta, Microsoft, Scotiabank, and VMware. Shadaj Laddad is supported in part by the NSF Graduate Research Fellowship Program under Grant No. DGE 2146752. Any opinions, findings, and conclusions or recommendations expressed in this material are those of the authors and do not necessarily reflect the views of the National Science Foundation.
\end{acks}
\bibliographystyle{ACM-Reference-Format}
\bibliography{bibliography}

\end{document}